\documentclass{article}
\usepackage[utf8]{inputenc}
\usepackage{amsmath,amsfonts,amssymb,amsthm,bm,graphicx,color}
\usepackage{caption,subfig,tikz}
\usepackage{geometry}
\usepackage{enumerate}
\usepackage[numbers,sort&compress]{natbib}
\usepackage{xcolor}
\usepackage{cleveref}

\definecolor{mygreen}{rgb}{0,0.7,0}

\newcommand{\dvol}{\operatorname{dVol}}
\newcommand{\vol}{\operatorname{Vol}}
\newcommand{\arcl}{\operatorname{arcl}}
\newcommand{\ind}{\mathbf{1}}
\newcommand{\RR}{\mathbb{R}}

\newcommand{\M}{\mathcal{M}}
\newcommand{\N}{\mathcal{N}}

\newcommand{\Mdim}{n}

\newcommand{\systemState}{\ensuremath{x}}
\newcommand{\parameterOne}{\ensuremath{\beta_1}}
\newcommand{\parameterTwo}{\ensuremath{\beta_2}}

\newcommand{\delayA}{{n}}
\newcommand{\delayB}{{n+1}}
\newcommand{\delayC}{{n+2}}

\newtheorem{thm}{Theorem}
\newtheorem{lem}[thm]{Lemma}
\newtheorem{Def}[thm]{Definition}

\newtheorem{cor}[thm]{Corollary}

\begin{document}
\title{A geometric approach \\
to the transport of discontinuous densities}

\author{Caroline Moosm\"uller\thanks{Department of Chemical and Biomolecular Engineering, Johns Hopkins University, Baltimore MD 21218, USA, Department of Mathematics, University of California, San Diego, La Jolla, CA 92093, USA. Email: \texttt{cmoosmueller$@$ucsd.edu}.}
\and Felix Dietrich\thanks{Department of Applied Mathematics and Statistics,
Department of Chemical and Biomolecular Engineering, Johns Hopkins University, Baltimore, MD 21218, USA. Email: \texttt{felix.dietrich$@$jhu.edu}, \texttt{yannisk$@$jhu.edu} (corresponding author).}
\and Ioannis G.\ Kevrekidis\footnotemark[2]
}

\date{}

\maketitle

\begin{abstract}
Different observations of a relation between inputs (``sources") and outputs (``targets") are often reported in terms of histograms (discretizations of the source and the target densities).
Transporting these densities to each other provides insight regarding the underlying relation. 
In (forward) uncertainty quantification, one typically studies how the distribution of inputs to a system affects the distribution of the system responses. Here, we focus on the identification of the system (the transport map) itself, once the input and output distributions are determined,
and suggest a modification of current practice by including data from what we call ``an observation process''.
We hypothesize that there exists a smooth manifold underlying the relation; the sources and the targets are then partial observations (possibly projections) of this manifold.
Knowledge of such a manifold implies knowledge of the relation, and thus of ``the right" transport between source and target observations.
When the source-target observations are not bijective (when the manifold is not the graph of a function over both observation spaces, either because folds over them give rise to density singularities, or because it marginalizes over several observables), recovery of the manifold is obscured.
Using ideas from attractor reconstruction in dynamical systems, we demonstrate how {\em additional information} in the form of short histories of an {\em observation process} can help us recover the underlying manifold. 
The types of additional information employed and the relation to optimal transport based solely on density observations is illustrated and discussed, along with limitations in the recovery of the true underlying relation. 


\end{abstract}



\section{Introduction}

\setlength{\unitlength}{5cm}

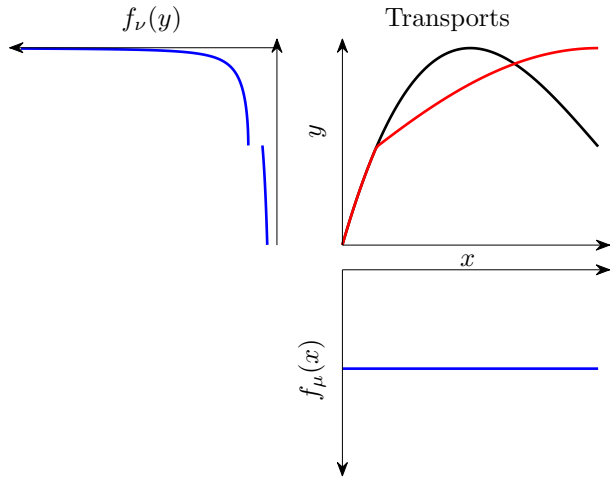
\begin{figure}[htp!]
\centering
\begin{picture}(1.3,1.3)
\put(-0.3,0){\includegraphics[scale=0.5]{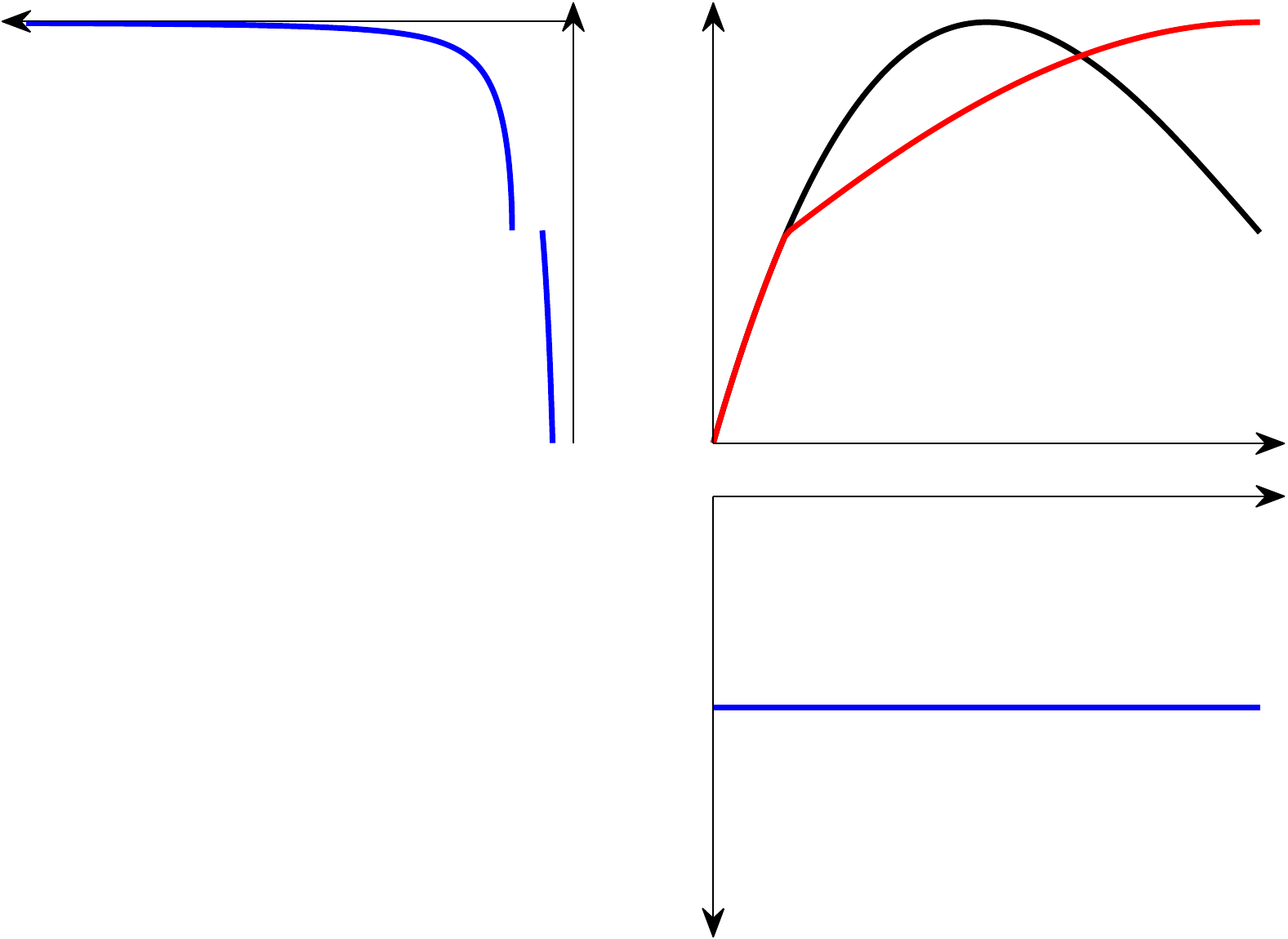}}
\put(0.47,0.2){\rotatebox{90}{$f_{\mu}(x)$}}
\put(0,1.2){$f_{\nu}(y)$}
\put(0.5,0.9){\rotatebox{90}{$y$}}
\put(0.9,0.56){$x$}
\put(0.7,1.2){Transports}
\end{picture}
\caption{The uniform density $f_{\mu}(x)$ on $[0,1]$ is transported by a ``folded'', i.e.\ non-injective transport (black, $y=T(x)= -2(1-x)^3 + 1.5(1-x) + 0.5$) to a density $f_{\nu}(y)$ which has a jump discontinuity and a singularity. The red transport map is the Wasserstein optimal transport between $f_{\mu}(x)$ and $f_{\nu}(y)$. In contrast to the black transport, the Wasserstein transport is one-to-one, but it is not $C^1$ everywhere.}
\label{fig:intro_hook}
\end{figure}

We are interested in studying observations of an (unknown), manifold when the histograms of the observed quantity suggest the existence of singularities.
Such cases arise in many physical or computational contexts:
e.g., observations of tracer diffusion on a non-flat cell surface \cite{adler-2019}, particle flow in MC simulations \cite{chen-2010}, or density of state singularities in carbon nanotubes~\cite{marchenko-2018}. 
\Cref{fig:intro_hook} illustrates how such observation 
singularities might arise: 
A uniform density of points $f_{\mu}(x)$ on a segment of the $x$-axis is 
non-injectively mapped, through the black parabola, to a segment of the $y$ axis giving rise to the density $f_{\nu}(y)$.
It is easy to see that this density has, by construction, both a jump discontinuity and
an ``infinity". 
The two densities ($f_{\mu}(x)$ and $f_{\nu}(y)$), constitute two different observations of the same manifold (the black curve). If we only know the two distributions, and do not know the underlying pointwise correspondences, we can attempt to transport them to each other using the Wasserstein optimal transport \cite{villani-2009,brenier-1991}.
Attempting this, for example with numerical algorithms as in \cite{cuturi-2013,parno-2016,parno-2018}, does not recover the manifold (the black curve in \Cref{fig:intro_hook}), but a different one, which is visibly non-differentiable at a single point (\Cref{fig:intro_hook}, red map). Mapping $x$-point observations through this red function to the $y$-axis results in the same $f_{\nu}(y)$, but pointwise different $y$-values, 
compared to the black curve that we know, in this case, to be ``the truth".

The need to create useful pointwise mappings (as opposed to only mappings of distributions) arises in several contexts (domain adaptation or transfer learning,
e.g. \cite{yair-2019,courty-2014}). 
This would be greatly facilitated if we could reconstruct, from our observations,
the {\em intrinsic manifold}, i.e., if we can find the ``intrinsic'' state, the point on the black curve. In \Cref{fig:intro_hook}, we can consider the ``intrinsic state" to be the arclength on the black curve, starting, without loss of generality, from one of its two ends.
Then both the ``source" ($x$) and the ``target" ($y$) observations are 
simply functions of this intrinsic state. 
We will discuss cases where it would make even less sense to use Wasserstein optimal transport, namely, when the dimensions of the source and target spaces differ: e.g.\ if one set of observations consists of scalars, while the other set is vector-valued.

Our goal is to devise conditions under which we can meaningfully reconstruct the original, underlying manifold. In general, starting with only the two distributions, this is not possible; so we must modify the problem.
A transport between two densities is reinterpreted as a geometric object (the function is identified with its graph). To embed the object, a possibility arises through the use of embedding theorems \cite{takens-1981,sauer-1991,whitney-1936}. 
To apply these theorems, we need to assume additional observation information, not just the density. This additional information can be provided through an {\em observation process} on the object, yielding observation histories, or, more generally, ``ensembles" of observations around each point.

With this additional information, one can create diffeomorphic copies of the intrinsic manifold from different sets of observations (e.g.\ short observation histories in $x$ and short observation histories in $y$). The densities on the two intrinsic manifold copies can then be easily transported, e.g by solving the Optimal Mass Transport (OMT) problem.

Another way to implement this transport is by using the Mahalanobis-like metric Diffusion Map framework \cite{coifman-2006,singer-2008,singer-2009,dsilva-2016}, which goes beyond diffeomorphisms to create {\em isometric} embeddings of the intrinsic manifold---the transport in this case reduces to a global orthogonal transformation \cite{berry-2016,kemeth-2018}.
For this, we also need the additional observation process data. 

The ``enhanced" reconstruction of the intrinsic manifold that exploits these additional process data allows us to deal with overlapping densities of observations, 
arising through non-invertible maps, see also \cite{dietrich-2018}. These include densities with singularities, therefore providing a contrasting approach to \cite{young-2016,loeper-2009,ma-2005} in the spirit of \cite{kevrekidis-2017}.
The approach also allows us to usefully transport marginal distributions if sufficient observation process data is available.

Even when the original function is invertible, the Wasserstein optimal transport might not give the ``correct" intrinsic solution, but one that can be transformed to it through a measure-preserving map, see e.g.\ McCann's polar factorization \cite{mccann-2001}. The Mahalanobis framework provides an isometry to the intrinsic manifold, so that the missing transformation to the solution is just a global orthogonal map.
This is particularly interesting in higher dimensions; on $\RR$ these concepts agree.

The paper is organized as follows: \Cref{sec:notation} introduces basic mathematical concepts concerning optimal transport and embedology. In \Cref{sec:one_dim_transport} we illustrate the recovery of manifolds and corresponding transport maps through time-delay embeddings in the one-dimensional setting. In \Cref{sec:densities in r2} we extend these results to the two-dimensional case, including a discussion of marginal distributions.
\Cref{sec:mahala_wasser} demonstrates the recovery of the intrinsic manifold up to an isometry. We compare this to the reconstruction with optimal transport, emphasizing the difference between measure-preserving and geometry-preserving maps.

\section{Basic mathematical concepts and notation}\label{sec:notation}
\subsection{Transport of densities}\label{sec:transport}
The optimal transport problem has been proposed by Monge \cite{monge-1781}, and seeks to find a mapping from one distribution of mass to another such that a cost function is minimized among all measure-preserving maps. We here introduce the mathematical background on $\RR^n$ (the general Riemannian manifold case is discussed in \Cref{sec:optimal transport on riemannian manifolds}), mainly following \cite{villani-2009}.

On $\RR^n$ we consider two measures $\mu$ and $\nu$. We say that $\nu$ is the push-forward of $\mu$ under $T: \RR^n \to \RR^n$, written as $\nu = T_{\sharp}\mu $, if
$
 \nu(A)=\mu(T^{-1}(A)), A \subset \RR^n,
$ 
where $T^{-1}(A)$ is the preimage of $A$ under $T$.
If $\mu \ll \lambda$ and $\nu \ll \lambda$, i.e.\ both $\mu$ and $\lambda$ are absolutely continuous with respect to the Lebesgue measure $\lambda$ on $\RR^n$, then there exist \emph{densities}, that is, Lebesgue-integrable functions $f_{\mu},f_{\nu}:\RR^n \to \RR$ such that
\begin{equation}\label{eq:absolut_cont}
    \mu(A) = \int_A f_{\mu}(x)\, d\lambda(x) \quad \text{and} \quad
    \nu(A) = \int_A f_{\nu}(y)\,d\lambda(y).
\end{equation}
A short-hand notation for \eqref{eq:absolut_cont} is $\mu = f_{\mu}\, \lambda$, and similarly for $\nu$.
The push-forward relation $\nu = T_{\sharp}\mu $ can then be rewritten as
\begin{equation}\label{eq:push_density}
    \int_A f_{\nu}(y)\,d\lambda(y) = \int_{T^{-1}(A)}f_{\mu}(x)\, d\lambda(x).
\end{equation}
If $T$ is invertible and differentiable, by change of variables, we can formulate \eqref{eq:push_density} as
\begin{equation}\label{eq:bijective_transport}
    f_{\nu}(y)=f_{\mu}\left(T^{-1}(y)\right)|\det D_yT^{-1}|.
\end{equation}
Given two densities $f_{\mu},f_{\nu}$, there might exist many transport maps satisfying \eqref{eq:bijective_transport}. Thus one often seeks to find a transport that in addition to \eqref{eq:bijective_transport} is also unique in some sense---typically represented by an \emph{optimization} problem. In the theory of optimal transport \cite{villani-2009}, the map $T$ is required to minimize a cost function of the form
\begin{equation}\label{eq:monge}
    \int_{\RR^n} c(x,T(x)) f_{\mu}(x)\,d\lambda(x),
\end{equation}
under the constraint \eqref{eq:bijective_transport}.
The cost $c$ is usually set to $c(x,y) = |x-y|^p$, and $p$ is typically set to $2$. The optimization of the Wasserstein problem with $p=2$ has a unique solution under reasonable assumptions \cite{villani-2009,brenier-1991}.

\subsection{Embedology}
In this paper, we use a dynamic observation process to provide additional information that will help us reconstruct intrinsic manifolds and define useful transport maps. The corresponding theory is broadly used when observing the time evolution of dynamical systems, and reconstructing attractors (long-term dynamics) through time series data.
In that context, Packard et al.~\cite{packard-1980},  Aeyels~\cite{aeyels-1981}, and Takens~\cite{takens-1981} describe the observability of nonlinear state spaces, by using a number of delays of a single, real-valued function of the system state to reconstruct the system attractor.
The results are mainly based on the theorems of Whitney~\cite{whitney-1936} (see appendix), which describe conditions for smooth functions on a large class of manifolds to construct embeddings.
Sauer, Yorke, and Casdagli~\cite{sauer-1991} later refined the concept by proving that ``almost all'' smooth functions can be used to construct embeddings into Euclidean space. The authors also generalized the results to fractal sets rather than smooth manifolds.
Here, we employ these embedding theorems as a theoretical underpinning for the reconstruction of intrinsic manifolds and transport maps.

\section{One dimensional transport and discontinuous densities}
\label{sec:one_dim_transport}

In this section, we illustrate the main ideas of the paper in the simple one-dimensional case ($n=1$), where the intrinsic manifold is the graph of a function over a single, real variable $x$. In later sections (\Cref{sec:densities in r2,sec:appendix proofs}), we describe the general setting in which the coordinates $x,y$ are considered functions over a (higher-dimensional) manifold.

In one dimension, there exists a unique, monotonically increasing, solution to the minimization of the cost function \eqref{eq:monge} under the constraint \eqref{eq:push_density}. 
It is given by
\begin{equation}\label{eq:transport_1D}
    W(x) =\left( F_{\nu}^{-1}\circ F_{\mu}\right)(x).
\end{equation}
Here $F_{\mu}$ is the cumulative distribution function (cdf) of $f_{\mu}$, defined by
%
$
    F_{\mu}(x)=\int_{-\infty}^{x}f_{\mu}(t) d\lambda(t),
$    
and similarly for $F_{\nu}$.
Our only restriction on the densities $f_{\mu}, f_{\nu}$ is that they be Lebesgue-integrable. Therefore, they may have discontinuities and singularities. In our illustrative example (\Cref{fig:intro_hook}) we consider a jump discontinuity and divergence to infinity. 
\begin{Def}
A Lebesgue-integrable density $f_{\nu}$ possesses a \emph{jump discontinuity} at $a\in \RR$ if
\begin{equation*}
    \lim_{y \to a ^ -}f_{\nu} (y) \neq \lim_{y \to a ^ +}f_{\nu} (y) \quad \text{and} \quad
    \lim_{y \to a ^ -}f_{\nu} (y),\lim_{y \to a ^ +}f_{\nu} (y) \in \RR.
\end{equation*}
A Lebesgue-integrable density $f_{\nu}$ diverges to infinity at $a\in \RR$ if
\begin{equation*}
    \lim_{y \to a ^ -}f_{\nu} (y) = \infty \quad \text{or} \quad
\lim_{y \to a ^ +}f_{\nu} (y) = \infty.
\end{equation*}
\end{Def}
\noindent Other types of discontinuities (e.g.\ derivative discontinuities) are also possible but not discussed here.

\subsection{Discontinuities and non-bijective transport}

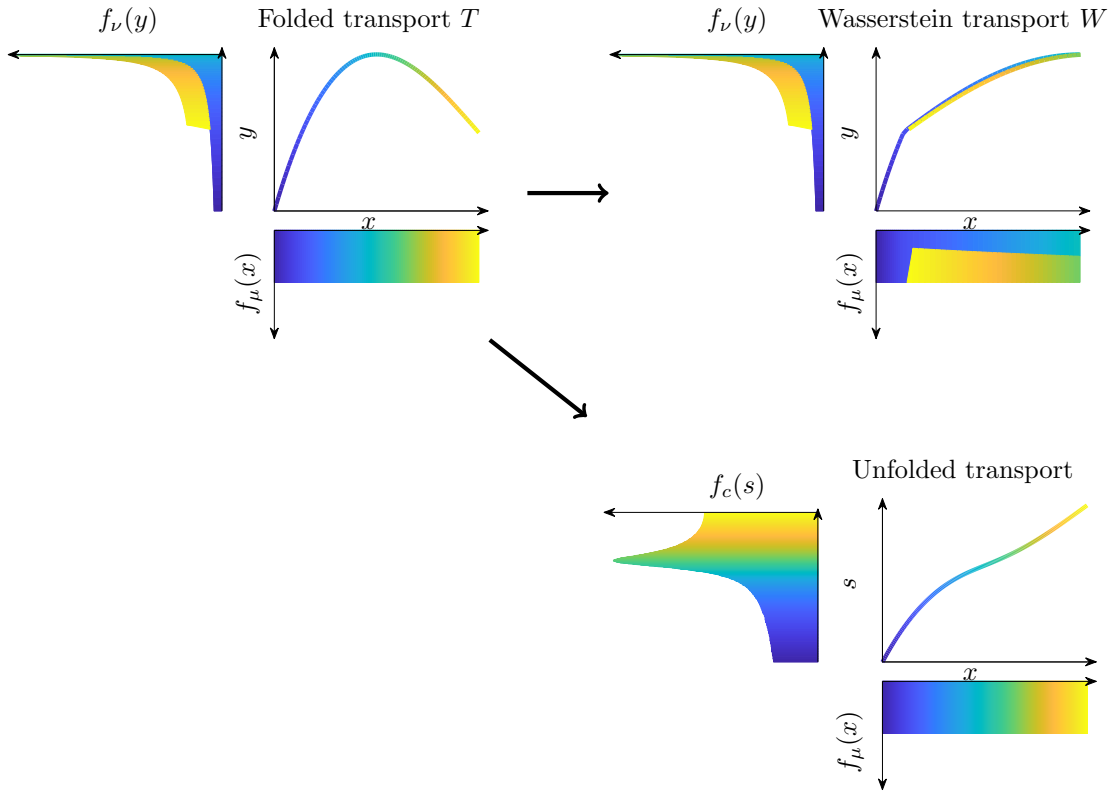
\begin{figure}[t!]
\hspace{-0.2cm}
\begin{tikzpicture}
\node[inner sep=0pt] (folded) at (0,5)
   {\includegraphics[scale=0.4]{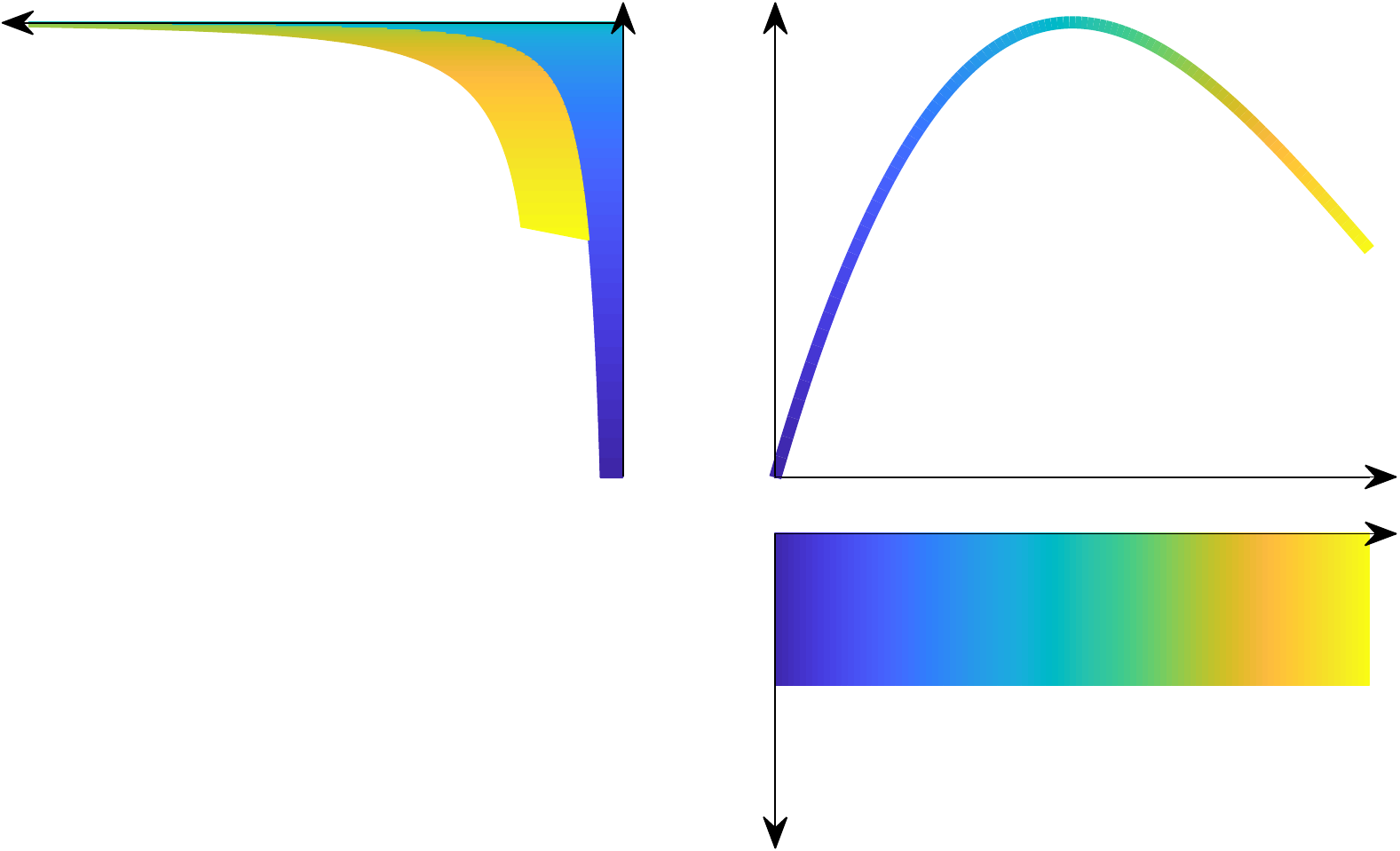}};
\node[inner sep=0pt] (Wasser) at (8,5)
   {\includegraphics[scale=0.4]{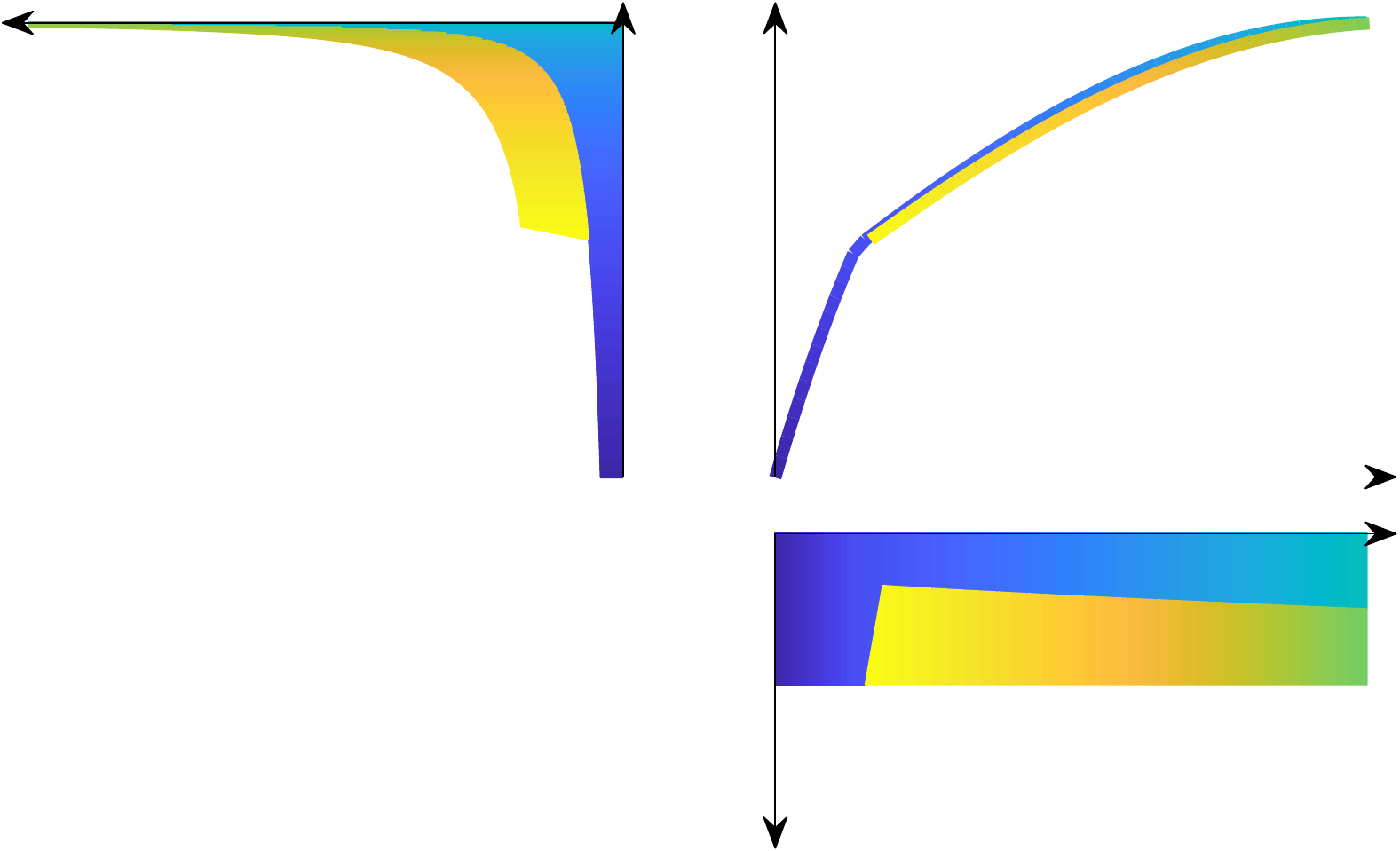}};
\node[inner sep=0pt] (unfold) at (8,-1)
   {\includegraphics[scale=0.4]{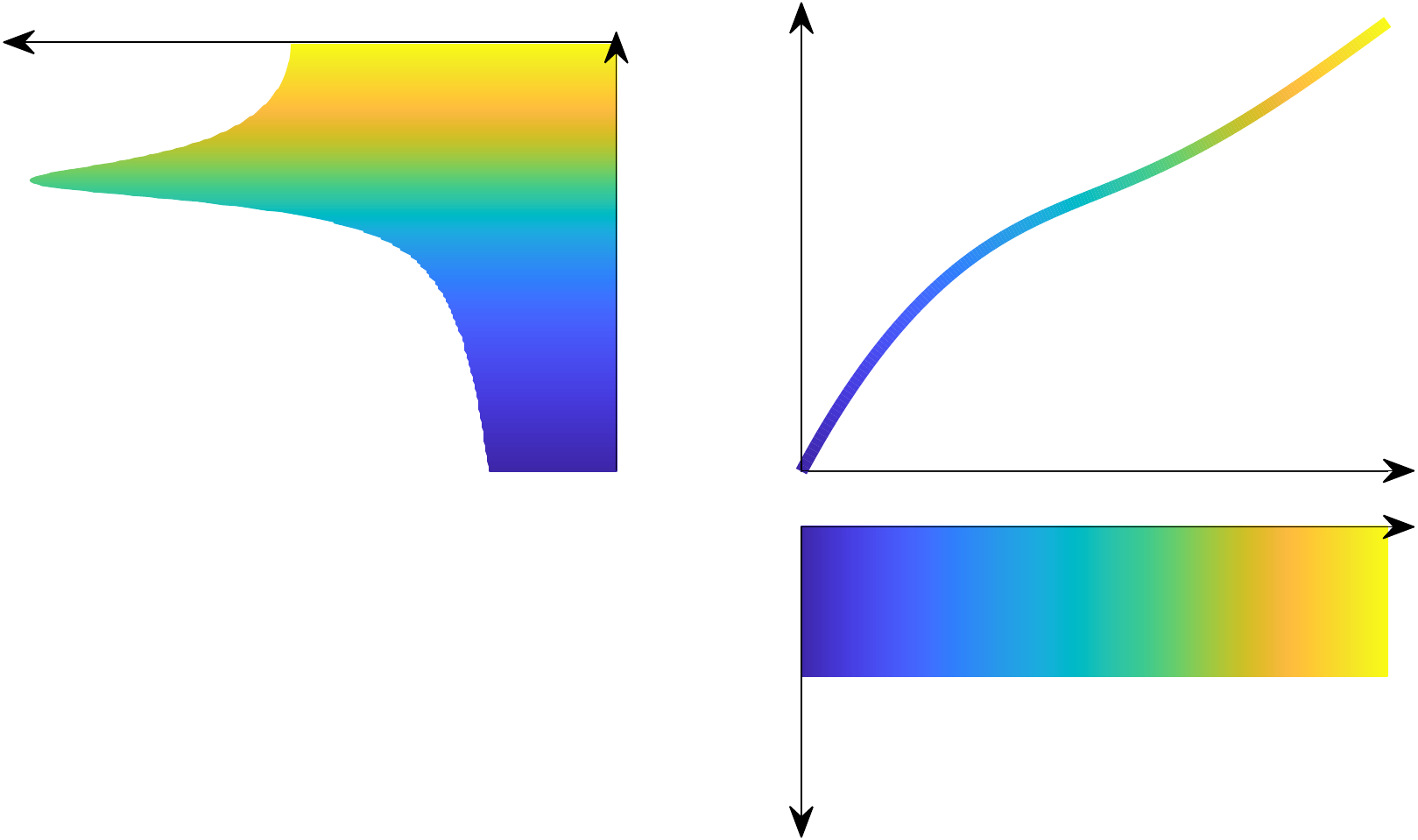}};
\node[inner sep=0pt] (unfold_aux) at (4.5,2){};
\node[inner sep=0pt] (folded_title) at (1.6,7.3){Folded transport $T$};
\node[inner sep=0pt] (Wasser_title) at (9.5,7.3){Wasserstein transport $W$};
\node[inner sep=0pt] (unfolded_title) at (9.5,1.3){Unfolded transport};
\node[inner sep=0pt] (x1) at (1.6,4.6){$x$};
\node[inner sep=0pt] (x2) at (9.6,4.6){$x$};
\node[inner sep=0pt] (x3) at (9.6,-1.4){$x$};
\node[inner sep=0pt] (y1) at (0,5.8){\rotatebox{90}{$y$}};
\node[inner sep=0pt] (y2) at (8,5.8){\rotatebox{90}{$y$}};
\node[inner sep=0pt] (s) at (8,-0.2){\rotatebox{90}{$s$}};
\node[inner sep=0pt] (fy1) at (-1.6,7.3){$f_{\nu}(y)$};
\node[inner sep=0pt] (fy2) at (6.5,7.3){$f_{\nu}(y)$};
\node[inner sep=0pt] (fs) at (6.5,1.1){$f_{c}(s)$};
\node[inner sep=0pt] (fx1) at (0,3.8){\rotatebox{90}{$f_{\mu}(x)$}};
\node[inner sep=0pt] (fx1) at (8,3.8){\rotatebox{90}{$f_{\mu}(x)$}};
\node[inner sep=0pt] (fx1) at (8,-2.3){\rotatebox{90}{$f_{\mu}(x)$}};
\draw[->,ultra thick,bend angle=0, bend left] (folded.south east) to (unfold_aux.north);
\draw[->,ultra thick,bend angle=0, bend left] (folded.east)++(0.5,0) to (Wasser.west);
\end{tikzpicture}
    \caption{\emph{Top left}: A uniform density $f_{\mu}(x)$ on $[0,1]$ is transported to a discontinuous density $f_{\nu}(y)$ with a smooth, non-injective transport (denoted by ``folded transport'' $T$).
    The coloring is with respect to the $x$-axis, and shows how $f_{\nu}(y)$ arises as the sum of the two branches of the transport map (compare \eqref{eq:non_inj_transport}).    
    \emph{Top right}:  The same uniform density $f_{\mu}(x)$ is transported to the same discontinuous density $f_{\nu}(y)$ with a continuous, but not $C^1$ transport (Wasserstein transport  $W$, computed with \eqref{eq:transport_1D}).
    The coloring schematically shows how the the two parts of the density in $y$ are mapped back to $x$. \emph{Bottom}: A uniform density $f_{\mu}(x)$ on $[0,1]$ is transported to $f_{c}(s)$, which is an ``unfolded'' version of $f_{\nu}(y)$. It is defined on the arclength $s$ of the curve arising from a time-delay embedding of observations of $f_{\nu}(y)$, described in \Cref{sec:unfold_discont_dens}. The ``unfolded transport'' is invertible.
    }
    \label{fig:joint_density}
\end{figure}

Consider a density $f_{\nu}$ which is pushed by a transport $T$ from a continuous density $f_{\mu}$. 
The two types of discontinuity defined above are illustated in \Cref{fig:intro_hook}. In this figure, the uniform density $f_{\mu}$ on the $x$-axis is transported to $f_{\nu}$ on the $y$-axis. 
We present two transports that push $f_{\mu}$ to $f_{\nu}$: The ``folded'' (i.e.\ non-injective) smooth map (black), and the Wasserstein optimal transport \eqref{eq:transport_1D} (red), which is continuous, but not $C^1$.
The divergence to infinity arises because both transport maps (black and red) have a maximum, i.e.\ satisfy $T'(x)=0$ at some $x \in [0,1]$.
The jump discontinuity of $f_{\nu}$ arises in each transport for a different reason:
\begin{enumerate}
    \item\label{transport_not_smooth} through the derivative discontinuity in the red transport map,
    \item\label{transport_not_bij} through the folding in the black transport map (the map is surjective, but not injective).
\end{enumerate}
Since we want to be able to work with noninvertible maps, like the black curve, we need to use \eqref{eq:push_density} rather than \eqref{eq:bijective_transport} to define push-forward of measures.
We now discuss case \eqref{transport_not_bij} in more detail. It illustrates that a continuous density can be pushed to a discontinuous one via a \emph{smooth} transport.
If the transport $T$ is not injective, there is a generalization of \eqref{eq:bijective_transport}:
\begin{equation}\label{eq:non_inj_transport}
    f_{\nu}(y)=\sum_{x \in T^{-1}(y)}\frac{f_{\mu}(x)}{|T'(x)|},
\end{equation}
if $T'(x)\neq 0$ for all $x$. Here $T^{-1}$ denotes the preimage.

Transporting a density with discontinuities as in \Cref{fig:intro_hook} to a uniform density with the Wasserstein optimal transport map $W$ \eqref{eq:transport_1D}, gives rise to a function which is not $C^1$ everywhere, as we integrate a discontinuous function, see \Cref{fig:joint_density} (top right).
Thus the non-injective transport $T$ of \Cref{fig:joint_density} (top left) is \emph{optimal} in the sense that it is smooth and transports the parametrization induced by the coloring correctly.
Note however that in general the Wasserstein cost \eqref{eq:monge} is smaller for $W$ than for $T$ ($W$ is closer to the identity map than $T$), hence $T$ is not optimal in the Wasserstein sense.
For simple examples similar to the one shown in \Cref{fig:joint_density}, the cost can be computed explicitly
\footnote{For a quadratic version of the folded transport $T$, given by $y=T(x)=-\left(\frac{3}{2}+\sqrt{2}\right)x^2+(2+\sqrt{2})x$, the cost \eqref{eq:monge} is $(1/60)(6+\sqrt{2})\approx 0.12$, which is more than for the Wasserstein optimal transport $W$, where the cost is $(1/60)(-54+41\sqrt{2}) \approx 0.07$.}.

In general, when dealing with a discontinuous density $f_{\nu}$ as in \Cref{fig:intro_hook}, we do not know that it consists of two (or more) branches as indicated by the coloring in \Cref{fig:joint_density}.
It is thus difficult to uncover the double folding and obtain the smooth transport $T$.
In the next sections we describe how the density $f_{\nu}$ can be unfolded by assuming some additional information (an observation process).
In \Cref{sec:C1_transport} we also suggest a construction of a $C^1$ transport only from the knowledge of the distributions, but without additional information.

\subsection{Unfolding discontinuous densities}\label{sec:unfold_discont_dens}

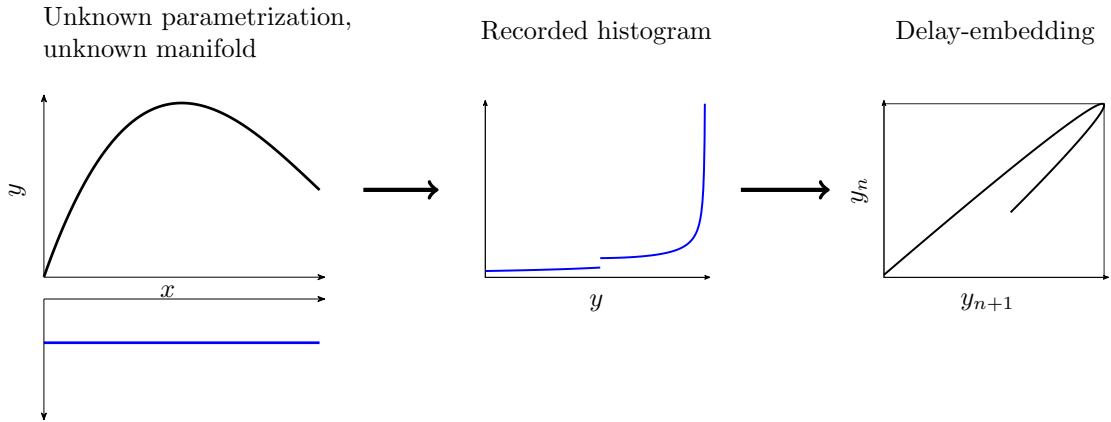
\begin{figure}[t!]
\centering
\begin{tikzpicture}
\node[inner sep=0pt,align=left] (caption_cusp) at (-1.4,2.1)
    {Unknown parametrization, \\unknown manifold};
\node[inner sep=0pt] (cusp) at (-1.6,-0.9)
   {\includegraphics[width=.25\textwidth]{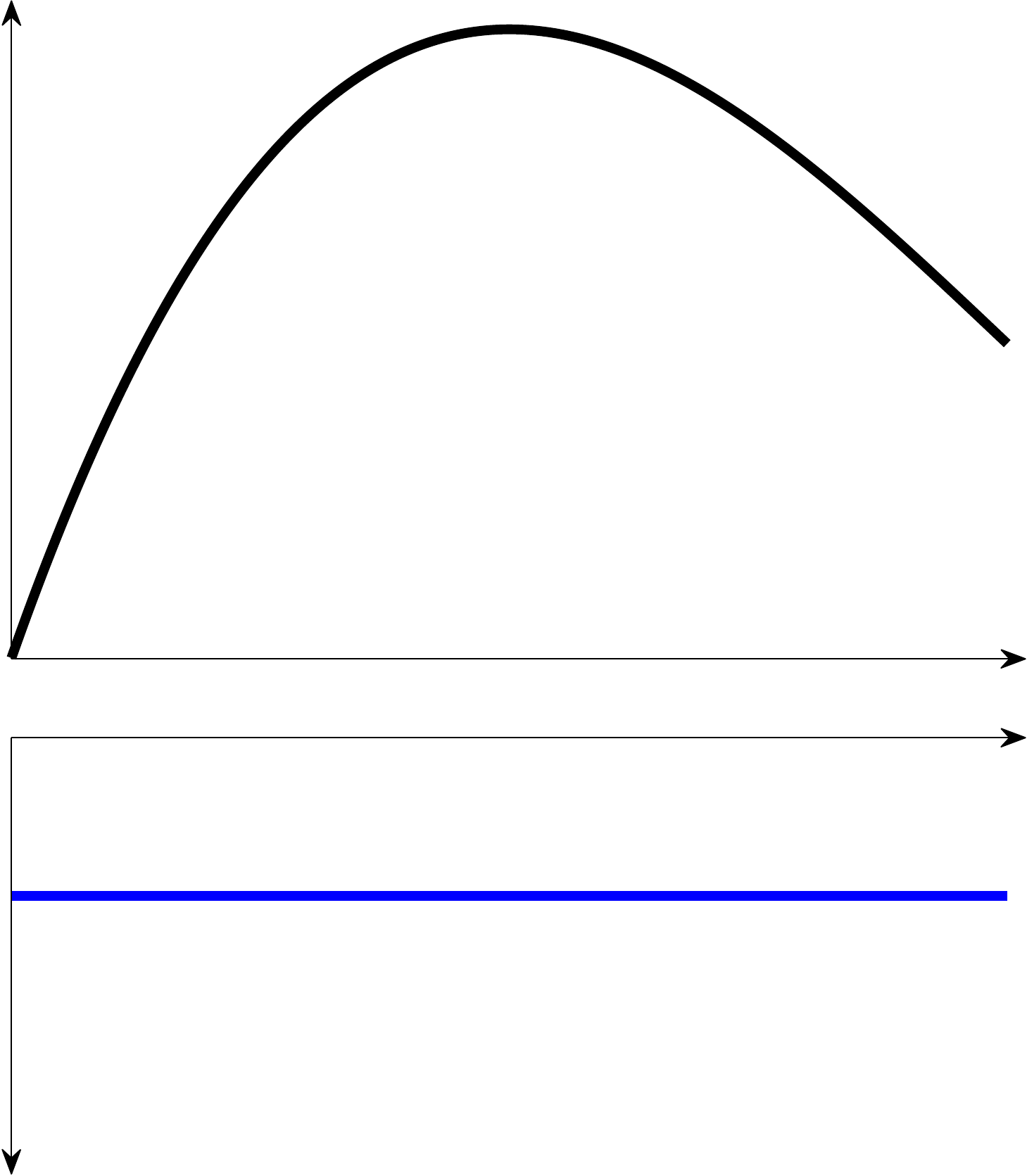}};
\node[inner sep=0pt] (y2) at (-3.8,0){\rotatebox{90}{$y$}};
\node[inner sep=0pt] (x2) at (-1.8,-1.35){$x$};
\node[inner sep=0pt] (caption_hist) at (3.9,2.1)
    {Recorded histogram};   
\node[inner sep=0pt] (mu) at (3.9,0)
    {\includegraphics[width=.2\textwidth]{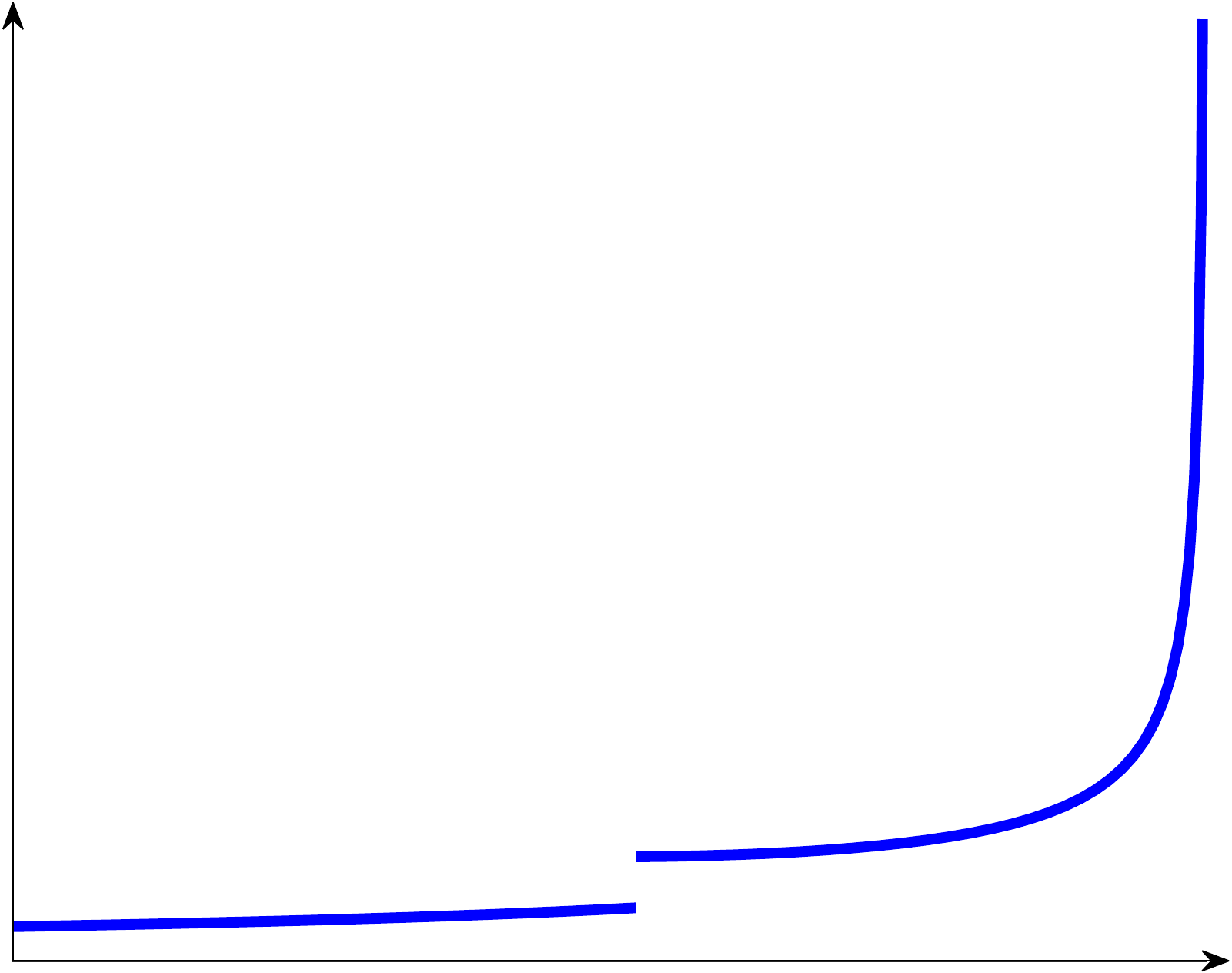}};
\node[inner sep=0pt] (y) at (3.9,-1.5){$y$};    
\node[inner sep=0pt] (mu_south) at (9,0){};
\node[inner sep=0pt] (mu_east) at (3,0){};
\node[inner sep=0pt] (cap_DMAP) at (9.2,2.1){Delay-embedding};
\node[inner sep=0pt] (DMAP) at (9.2,0)
    {\includegraphics[width=.2\textwidth]{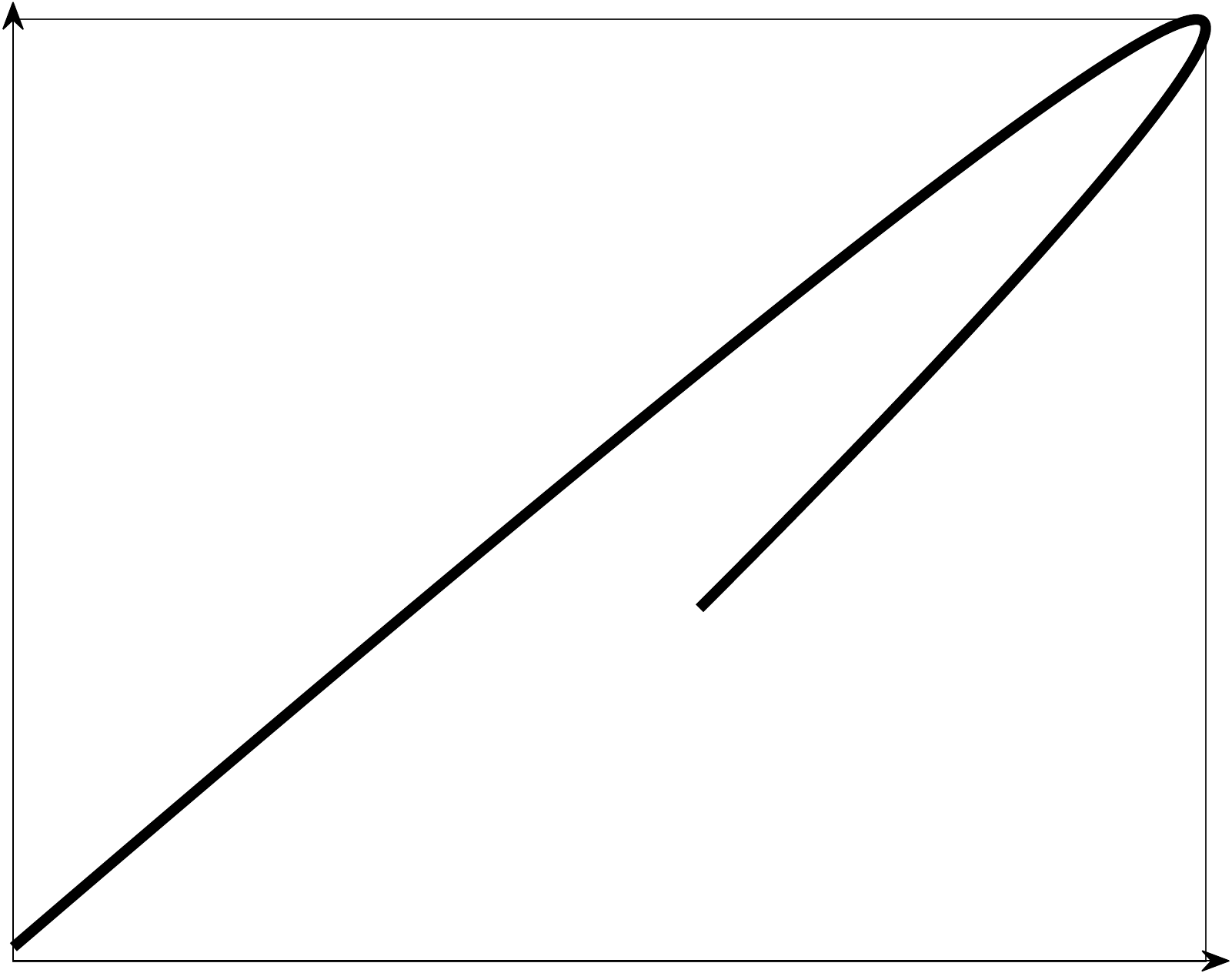}};
\node[inner sep=0pt] (phi1) at (9.1,-1.5){$y_{\delayB}$};
\node[inner sep=0pt] (phi2) at (7.4,0){\rotatebox{90}{$y_\delayA$}};
\draw[->,ultra thick,bend angle=0, bend left] (0.8,0) to (1.8,0);
\draw[->,ultra thick,bend angle=0, bend left] (mu.east)++(0.4,0) to (7,0);
\end{tikzpicture}
\caption{
\emph{Left}: On an unknown manifold (black curve), obtained from an unknown parametrization (uniform on $x$-axis, blue), we observe consecutive values of $y$ starting from uniformly distributed points in $x$. 
\emph{Middle}:
Histogram recorded from $y$-values.
\emph{Right:}
With the knowlegde of two consecutive $y$ values, we can
reconstruct (up to a diffeomorphism) the original curve using a delay-embedding. We show the time-delay embedding in $(y_{\delayB},y_\delayA)$. A parametrization of this curve (e.g.\ by arclength) can be used to transport the points to the original parametrization on $x$ (see \Cref{fig:joint_density}, bottom right).}
\label{fig:hook_overiew}
\end{figure}

We start by explaining the idea of unfolding discontinuous densities.
Given a curve in $\RR^2$, $c(s)=(c_1(s),c_2(s))$, and a density $f_{c}$ over $s$, we can push this density to the two axes by projecting it to the coordinates $c_1$ and $c_2$. Even in the case that both $c_1$ and $c_2$ are non-injective, we can push the density with the generalized formula \eqref{eq:non_inj_transport}.

A simple example is the graph of a function $T$, given by the curve
\begin{equation}\label{eq:transport_as_curve}
c(x)=(x,T(x)), \quad x \in [0,1].
\end{equation}
Here the projection to the first coordinate is the identity, and the projection to the second coordinate is given through the function $T$.
Now we can reparametrize the curve $c$ with a bijective map $s=\varphi(x)$ resulting in a curve
$
\tilde{c}(s)=c(\varphi^{-1}(s))=(\varphi^{-1}(s),T(\varphi^{-1}(s))).
$
We then consider a transport from $s$ to $x$.
\Cref{fig:joint_density} illustrates the effect of the coordinate transformation $\varphi$. The uniform density $f_{\mu}$ on $[0,1]$ ($x$-axis) is transported to a discontinuous density $f_{\nu}$ on the $y$-axis via a non-injective transport map $T$.
Reparametrizing the curve $c(x)=(x,T(x))$ by its arclength
\begin{equation}
    s = \varphi(x)= \arcl(x)=\int_0^x ||\dot{c}(t)||dt
\end{equation}
gives rise to the density
\begin{equation}\label{density_arclength}
    f_{c}(s)=f_{\mu}(\arcl^{-1}(s))|{\arcl^{-1}}'(s)|.
\end{equation}
The density $f_{c}$ on the arclength $s$ is continuous, in contrast to the discontinuous density $f_{\nu}$ on $y$ (\Cref{fig:joint_density}). This can also be seen by rewriting \eqref{density_arclength}:
\begin{equation*}
    f_{c}(s)=f_{\mu}(\arcl^{-1}(s))\frac{1}{\sqrt{1+T'(\arcl^{-1}(s))^2}},
\end{equation*}
where the denominator is never zero.
Projecting from $s$ to the second component of $\tilde{c}$, we obtain the transport from $f_{c}$ to $f_{\nu}$, which is the original transport $T$ up to the one-to-one transport between $x$ and $s$.
As we will see below, by using additional information from an observation process, we can recover the fact that $T$ is folded.

\subsection{Time-delay embedding}\label{sec:hook}

\setlength{\unitlength}{5cm}

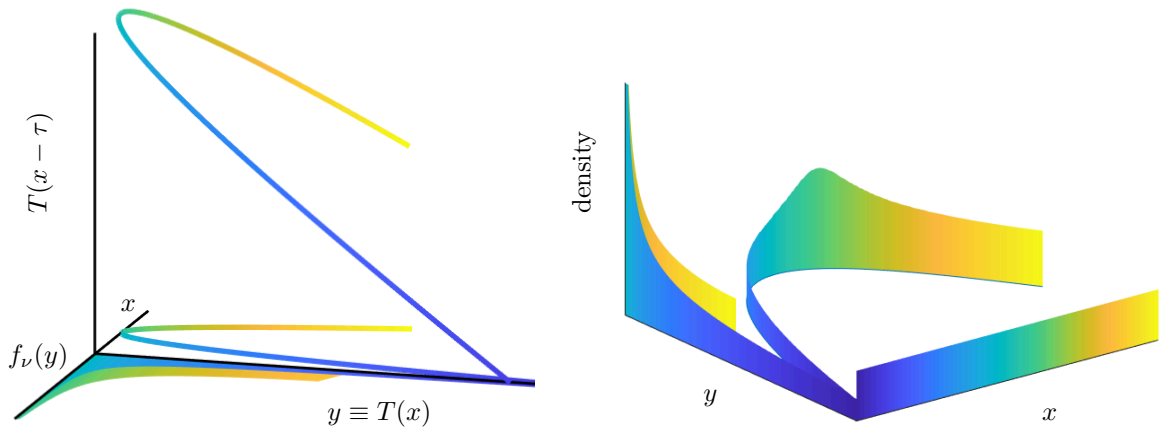
\begin{figure}[t!]
\centering
\begin{picture}(0.9,1.4)
\put(-1.3,-0.1){\includegraphics[scale=0.5]{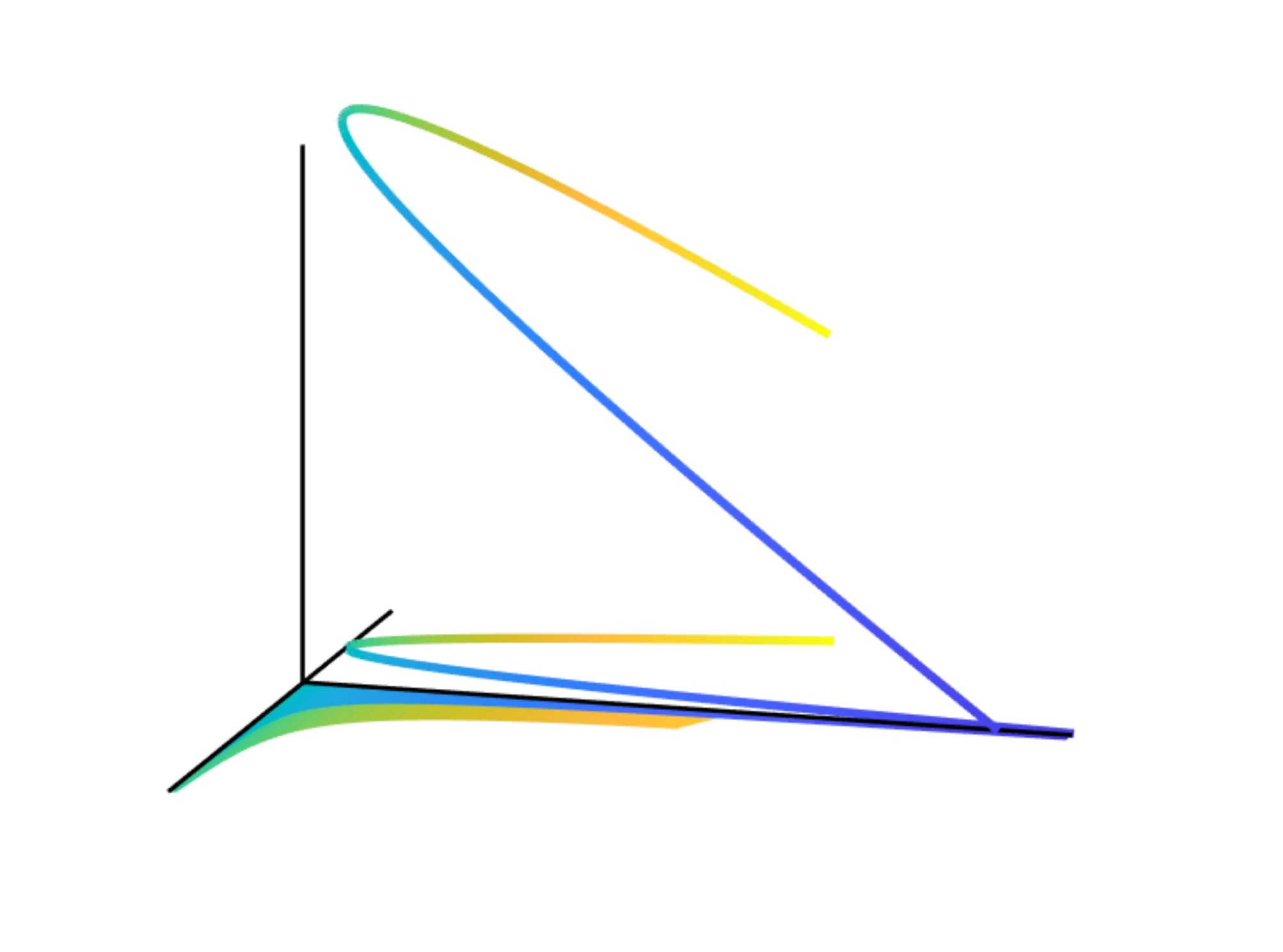}}
\put(-1,0.7){\rotatebox{90}{$T(x-\tau)$}}
\put(-0.75,0.44){$x$}
\put(-1.05,0.3){$f_{\nu}(y)$}
\put(-0.2,0.15){$y\equiv T(x)$}
\put(0.35,0){\includegraphics[scale=0.45]{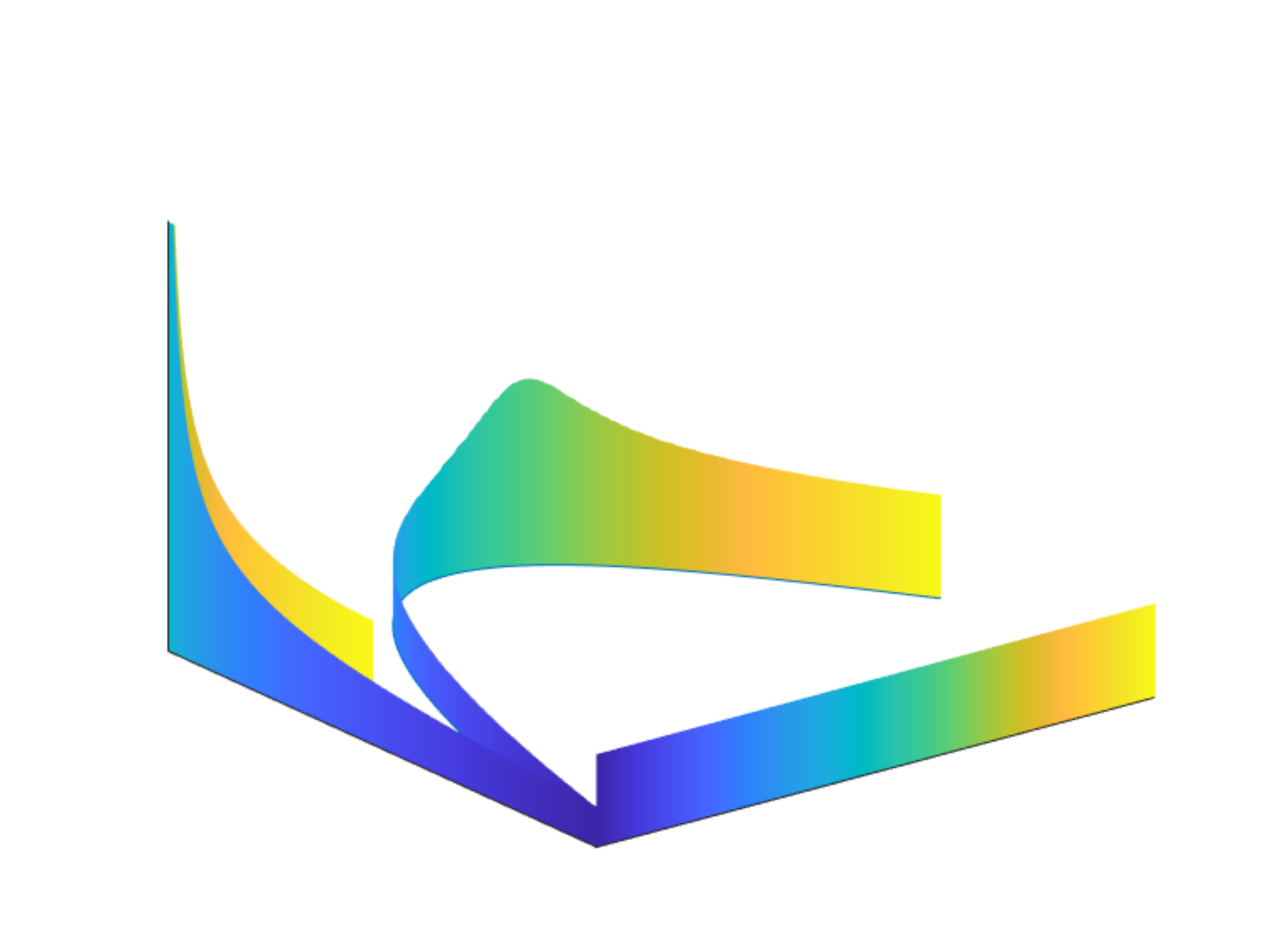}}
\put(1.7,0.15){$x$}
\put(0.8,0.2){$y$}
\put(0.45,0.7){\rotatebox{90}{density}}
\end{picture}
\caption{
\emph{Left:} Plot shows the curve $c(x)=(x,T(x))$ in the $(x,y,0)$-plane, the discontinuous density $f_{\nu}(y)$ on $y \equiv T(x)$, and the curve $x \mapsto (x,T(x),T(x-\tau))$, all colored by arclength of $d(x) = (T(x-\tau),T(x))$.
\emph{Right:} Densities are plotted as height over the respective axis: Uniform density on $[0,1]$ ($x$-axis), discontinuous density on $y$-axis, and density over the arclength, plotted on the curve (transport map).}
\label{fig:3D_hook}
\end{figure}

In the previous section we assumed the knowledge of the full curve $c(x)=(x,T(x))$ to explain the general concept. 

In this section, we show how equivalent results can be achieved by employing a (systematic but unknown) ``observation process"---considering, for example, that values of $y$ are always observed along trajectories of a systematic constant motion in $x$.
We do not assume that all the $y$ observations have been produced along a single constant-speed $x$ trajectory; we may have many short trajectories initialized at various random $x$ values.
Yet all of the observations follow the same process (the same, unknown, $x$-sampling dynamics).

Now the curve $c$ is embedded 
in a higher-dimensional (here, two-dimensional) space
by using the time-delay embedding $d$:
\begin{equation}\label{eq:time_delay_curve}
    d(x)=(T(x-\tau),T(x)),
\end{equation}
with time-delay parameter $\tau$.
By Takens' results \cite{takens-1981}, we know that such a time-delay embedding is diffeomorphic to the original curve (see \Cref{fig:hook_overiew}), if enough time delays and generic observations are used.
In this particular example we only need a single delay to obtain a diffeomorphic embedding (\Cref{fig:3D_hook}).
In general, to embed an $\Mdim$-dimensional manifold, $2\Mdim+1$ observation functions (e.g., one real-valued observation plus $2\Mdim$ delays) will be sufficient (see \Cref{sec:appendix proofs}).

As in \Cref{sec:unfold_discont_dens}, we can transport the uniform density $f_{\mu}(x)$ on $[0,1]$ to a density $f_{d}(s')$ along the arclength $s'$ of the curve embedded through time-delays \eqref{eq:time_delay_curve}. The density $f_{d}$ is shown on the original curve $c$ in \Cref{fig:3D_hook} (right).


With the same approach we can also recover the underlying manifold in more complicated examples. In \Cref{fig:hook_more_folds} we consider a function consisting of many folds (left). 
The histogram obtained by transporting the uniform density on $[0,1]$ ($x$-axis) with this function consists of many discontinuities (middle). By using delay coordinates $(y_{n+1},y_n)$ (and PCA) we can recover a manifold that is diffeomorphic to the original function.
A parametrization of this curve, here obtained with DMAP (right) can be used to transport to the uniform density on $[0,1]$ ($x$-axis).

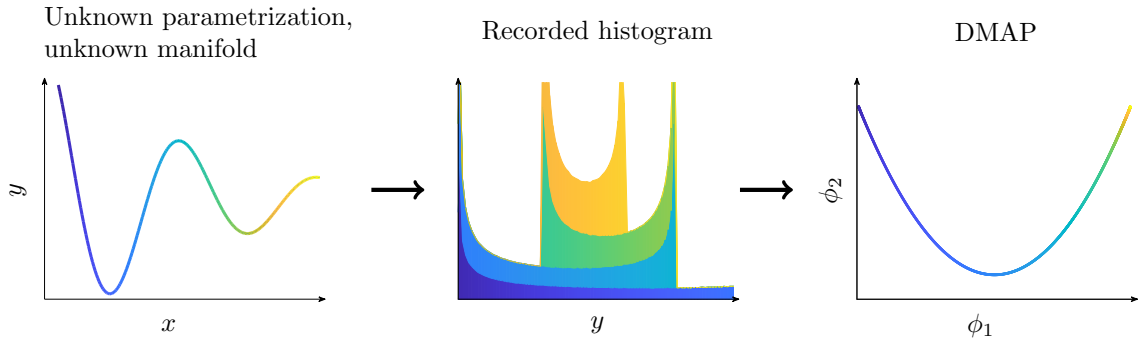
\begin{figure}[t!]
\centering
\begin{tikzpicture}
\node[inner sep=0pt,align=left] (caption_cusp) at (-1.4,2.1)
    {Unknown parametrization, \\unknown manifold};
\node[inner sep=0pt] (cusp) at (-1.6,0)
   {\includegraphics[width=.25\textwidth]{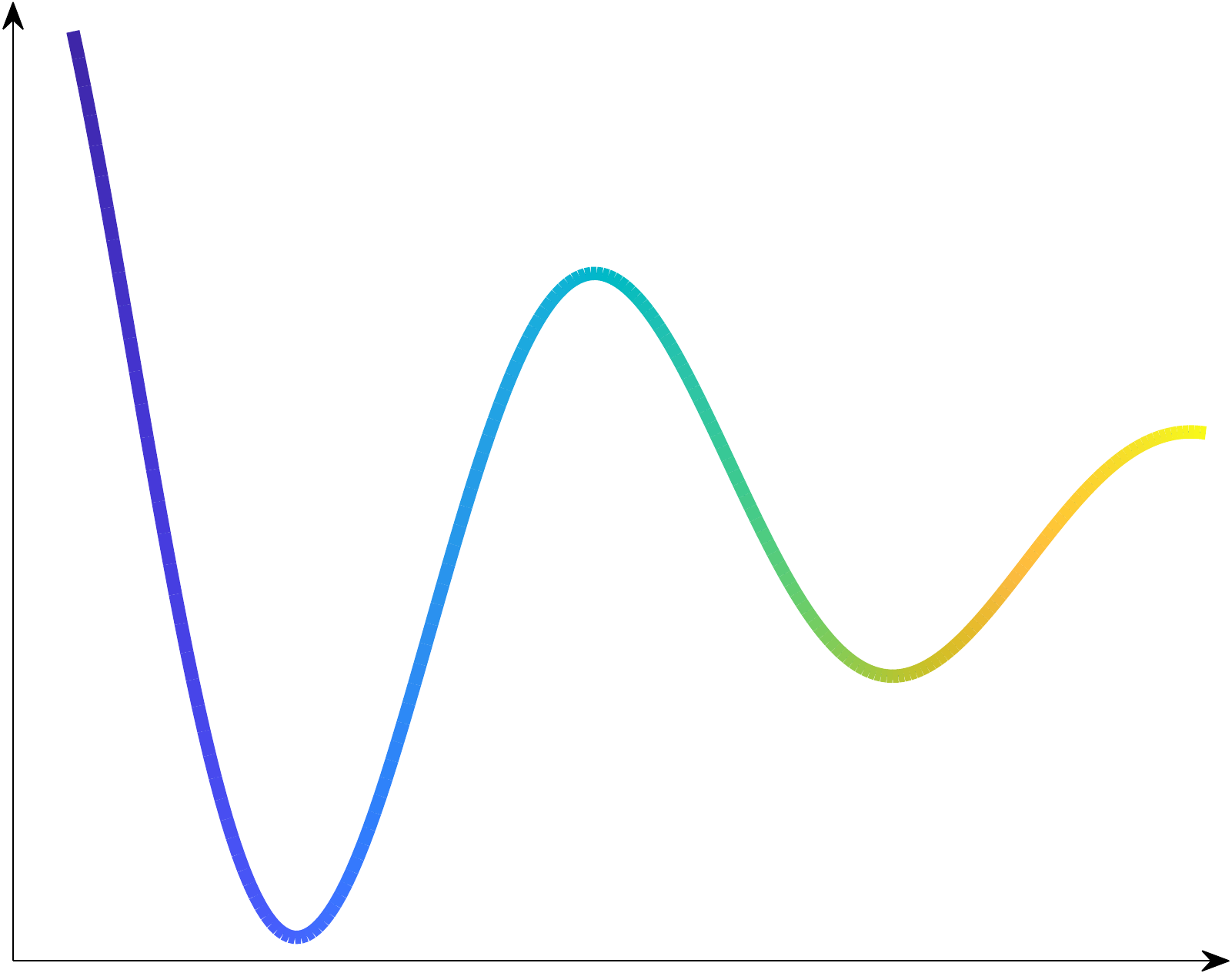}};
\node[inner sep=0pt] (y2) at (-3.8,0){\rotatebox{90}{$y$}};
\node[inner sep=0pt] (x2) at (-1.8,-1.8){$x$};
\node[inner sep=0pt] (caption_hist) at (3.9,2.1)
    {Recorded histogram};   
\node[inner sep=0pt] (mu) at (3.9,0)
    {\includegraphics[width=.25\textwidth]{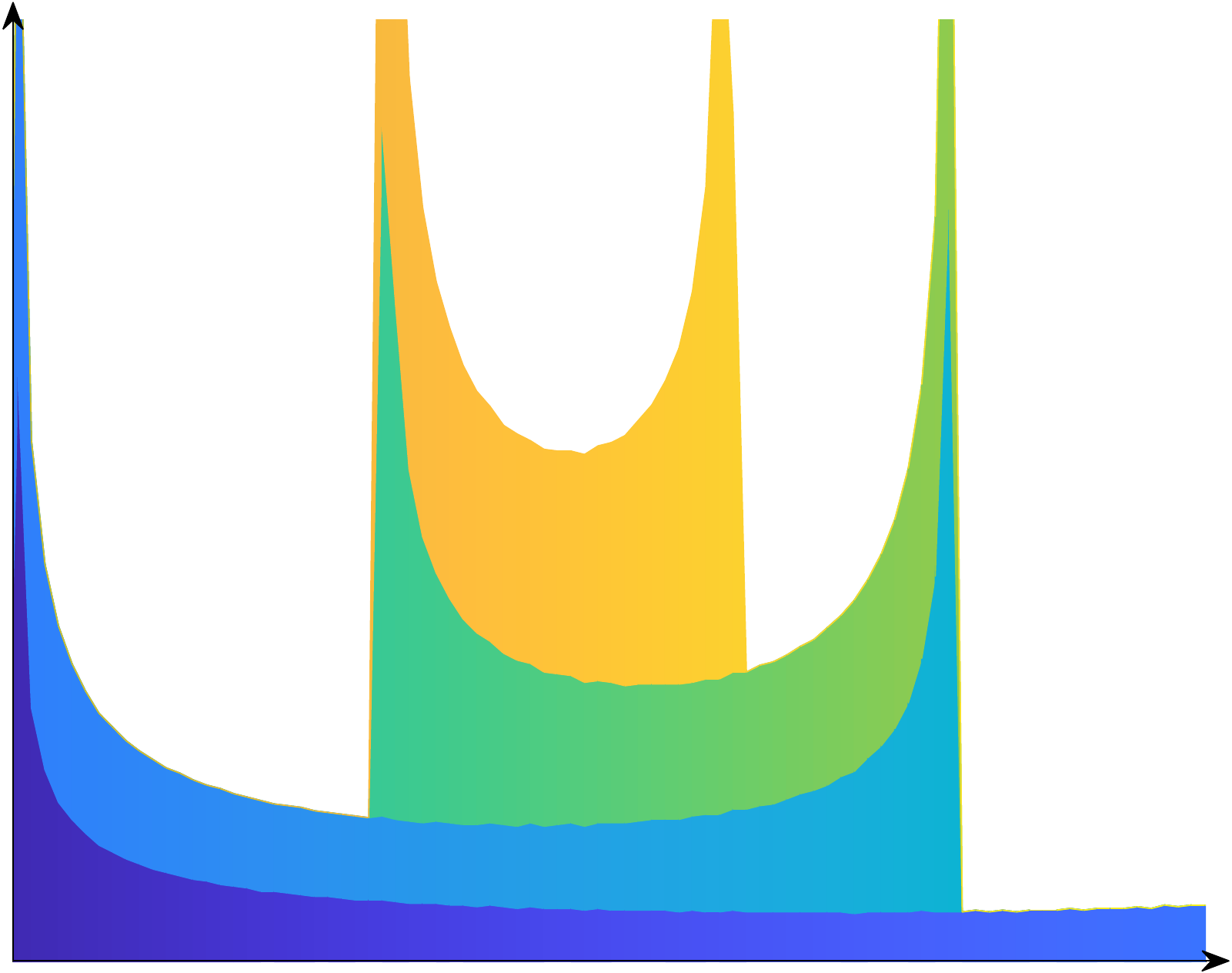}};
\node[inner sep=0pt] (y) at (3.9,-1.8){$y$};    
\node[inner sep=0pt] (mu_south) at (9,0){};
\node[inner sep=0pt] (mu_east) at (3,0){};
\node[inner sep=0pt] (cap_DMAP) at (9.2,2.1){DMAP};
\node[inner sep=0pt] (DMAP) at (9.2,0)
    {\includegraphics[width=.25\textwidth]{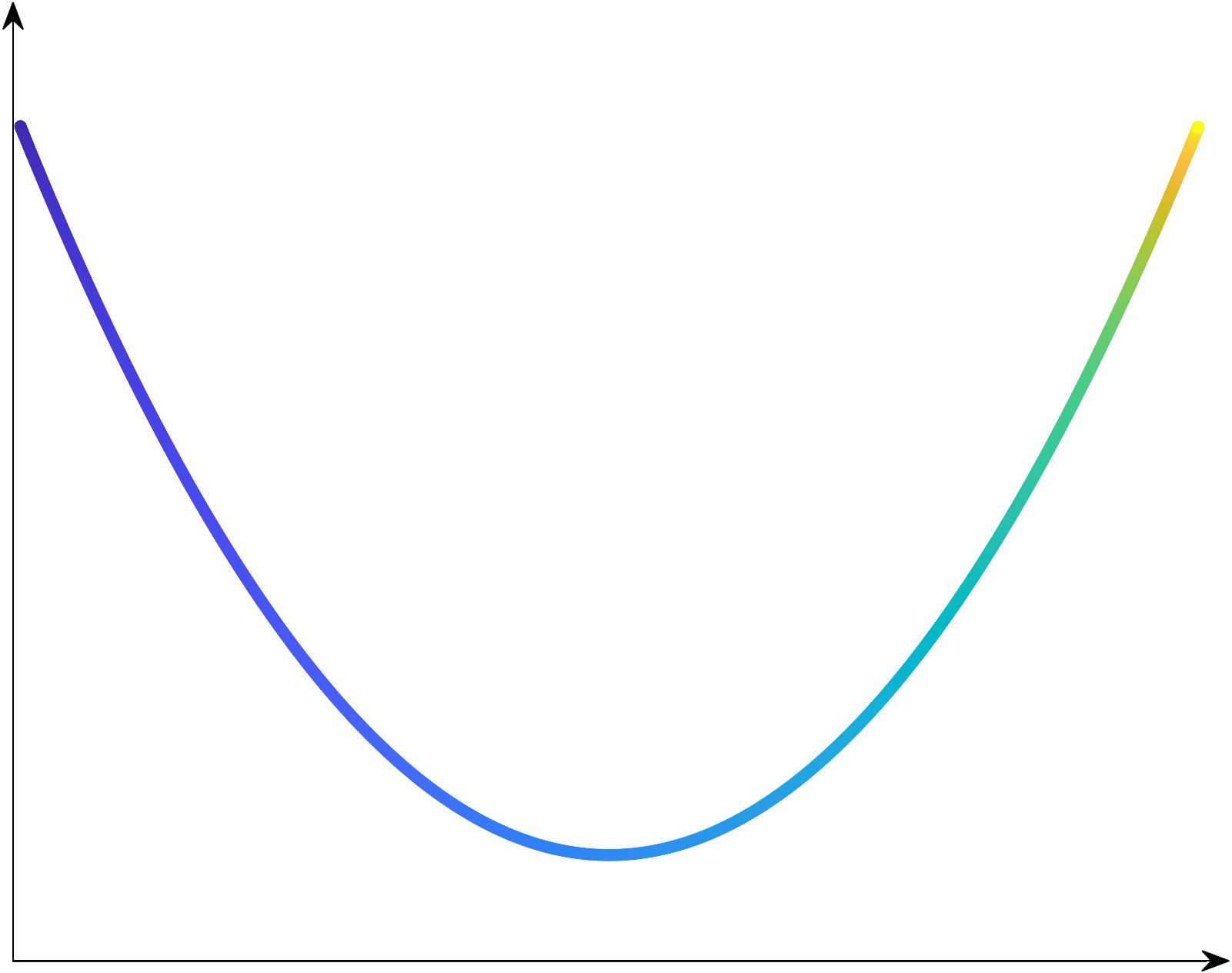}};
\node[inner sep=0pt] (phi1) at (9,-1.8){$\phi_1$};
\node[inner sep=0pt] (phi2) at (7,0){\rotatebox{90}{$\phi_2$}};
\draw[->,ultra thick,bend angle=0, bend left] (cusp.east)++(0.6,0) to (1.6,0);
\draw[->,ultra thick,bend angle=0, bend left] (mu.east) to (6.5,0);
\end{tikzpicture}
\caption{
\emph{Left}: On an unknown manifold (curve), obtained from an unknown parametrization (uniform on $x$-axis, indicated by coloring), we observe consecutive values of $y$ starting from uniformly distributed points in $x$. 
\emph{Middle}:
Histogram recorded from $y$-values. The coloring of the histogram indicates from which part of the curve (left) it has been produced.
\emph{Right:}
With the knowlegde of two consecutive $y$ values, we can
reconstruct (up to a diffeomorphism) the original curve using a delay-embedding. Here we show a
parametrization of this curve with DMAP embedding $(\phi_1,\phi_2)$ applied to PCA coordinates of $(y_{\delayB},y_\delayA)$. This embedding can be used to transport the points to the original uniform parametrization on $x$.}
\label{fig:hook_more_folds}
\end{figure}

Instead of functions, we can even consider relations: 
In \Cref{fig:curve2} the underlying manifold it is a closed curve.
In this case, the recorded histogram is obtained by pushing the uniform density on the arclength of the curve to the $y$-axis. Note that in these examples we can also consider the histogram obtained by pushing the uniform density on the arclength to the $x$-axis---in the examples where we consider functions (\Cref{fig:hook_overiew,fig:hook_more_folds}) this is not interesting, as we already parametrize the manifolds by the $x$-axis.

Also in the case of relations we can reconstruct a diffeomorphic copy of the underlying manifold with a time-delay embedding in $(y_{n+2},y_{n+1},y_n)$. As in the case for functions, a parametrization (obtained by, e.g.\ DMAP) can be used to transport to the uniform density.

To conclude: Diffeomorphic copies of one-dimensional curves can be constructed through time-delay embeddings. The densities on these curves are continuous, in contrast to the original target densities.

\begin{figure}[t!]
\centering
\begin{tikzpicture}
\node[inner sep=0pt,align=left] (caption_cusp) at (-1.4,2.1)
    {Unknown parametrization, \\unknown manifold};
\node[inner sep=0pt] (cusp) at (-1.6,0)
   {\includegraphics[width=.2\textwidth]{{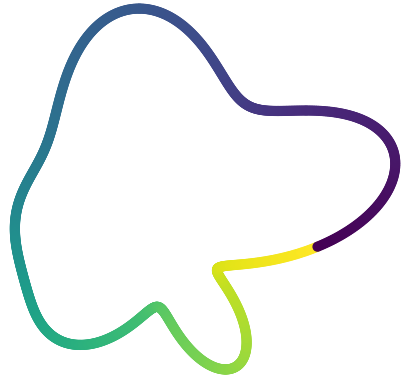}}};
\node[inner sep=0pt] (y2) at (-3.8,0){\rotatebox{90}{$y$}};
\node[inner sep=0pt] (x2) at (-1.8,-1.8){$x$};
\draw[-] (-3.5,-1.6) to (0.3,-1.6);
\draw[-] (-3.5,-1.6) to (-3.5,1.1); 
\node[inner sep=0pt] (caption_hist) at (3.9,2.1)
    {Recorded histogram};   
\node[inner sep=0pt] (mu) at (3.9,0)
    {\includegraphics[width=.2\textwidth]{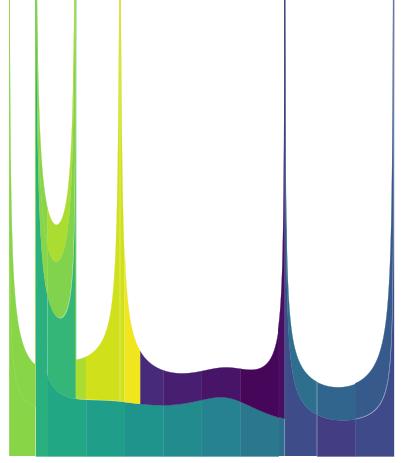}};
\node[inner sep=0pt] (y) at (3.9,-2){$y$};   
\draw[-] (2.3,-1.8) to (5.5,-1.8);
\draw[-] (2.3,-1.8) to (2.3,1.7); 
\node[inner sep=0pt] (mu_south) at (9,0){};
\node[inner sep=0pt] (mu_east) at (3,0){};
\node[inner sep=0pt] (cap_DMAP) at (9.2,2.1){DMAP};
\node[inner sep=0pt] (DMAP) at (9.2,0)
    {\includegraphics[width=.2\textwidth]{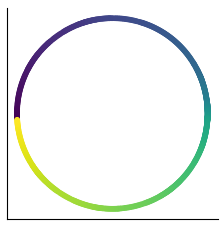}};
\node[inner sep=0pt] (phi1) at (9.2,-1.8){$\phi_1$};
\node[inner sep=0pt] (phi2) at (7.3,0){\rotatebox{90}{$\phi_2$}};
\draw[->,ultra thick,bend angle=0, bend left] (cusp.east)++(0.6,0) to (1.6,0);
\draw[->,ultra thick,bend angle=0, bend left] (mu.east)++(0.3,0) to (6.9,0);
\end{tikzpicture}
\caption{
\emph{Left}: On an unknown manifold (curve), obtained from an unknown parametrization (uniform on its arclength), we observe consecutive values of projections to the $y$-axis starting from uniformly distributed points on the arclength of the curve. 
\emph{Middle}:
Histogram recorded from $y$-values. The coloring of the histogram indicates from which part of the curve (left) it has been produced.
\emph{Right:}
With the knowlegde of three consecutive $y$ values, we can
reconstruct (up to a diffeomorphism) the original curve using a delay-embedding.
Here we show a
parametrization of this curve with DMAP embedding $(\phi_1,\phi_2)$ applied to PCA coordinates of $(y_{\delayC},y_{\delayB},y_\delayA)$. This embedding can be used to transport the points to the original uniform parametrization on the arclength of the curve.}
\label{fig:curve2}
\end{figure}
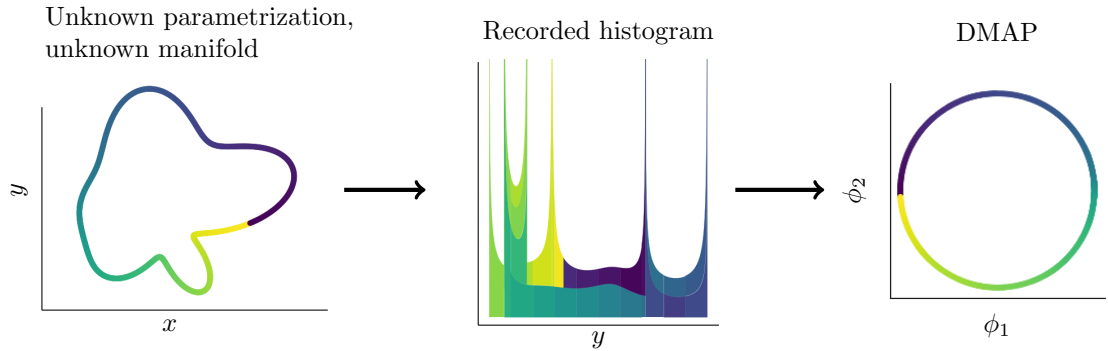
\section{Densities on two dimensional manifolds}\label{sec:densities in r2}
In two dimensions, we illustrate cases exhibiting one-parameter families of discontinuities and study marginals of two-dimensional distributions.
We show how the idea of ``unfolding" discontinuities through process observations (in the form of delays) can be applied to such singular densities, and how the same approach can help construct joint distributions from marginal ones; both cases involve the construction of a copy of an underlying manifold from process observation histories.

\subsection{Unfolding two-dimensional discontinuous densities}
\begin{figure}[t!]
\centering
\begin{tikzpicture}
\node[inner sep=0pt] (caption_cusp) at (9,9.5)
    {Unknown parametrization, unknown manifold};
\node[inner sep=0pt] (cusp) at (9,7)
   {\includegraphics[width=.3\textwidth]{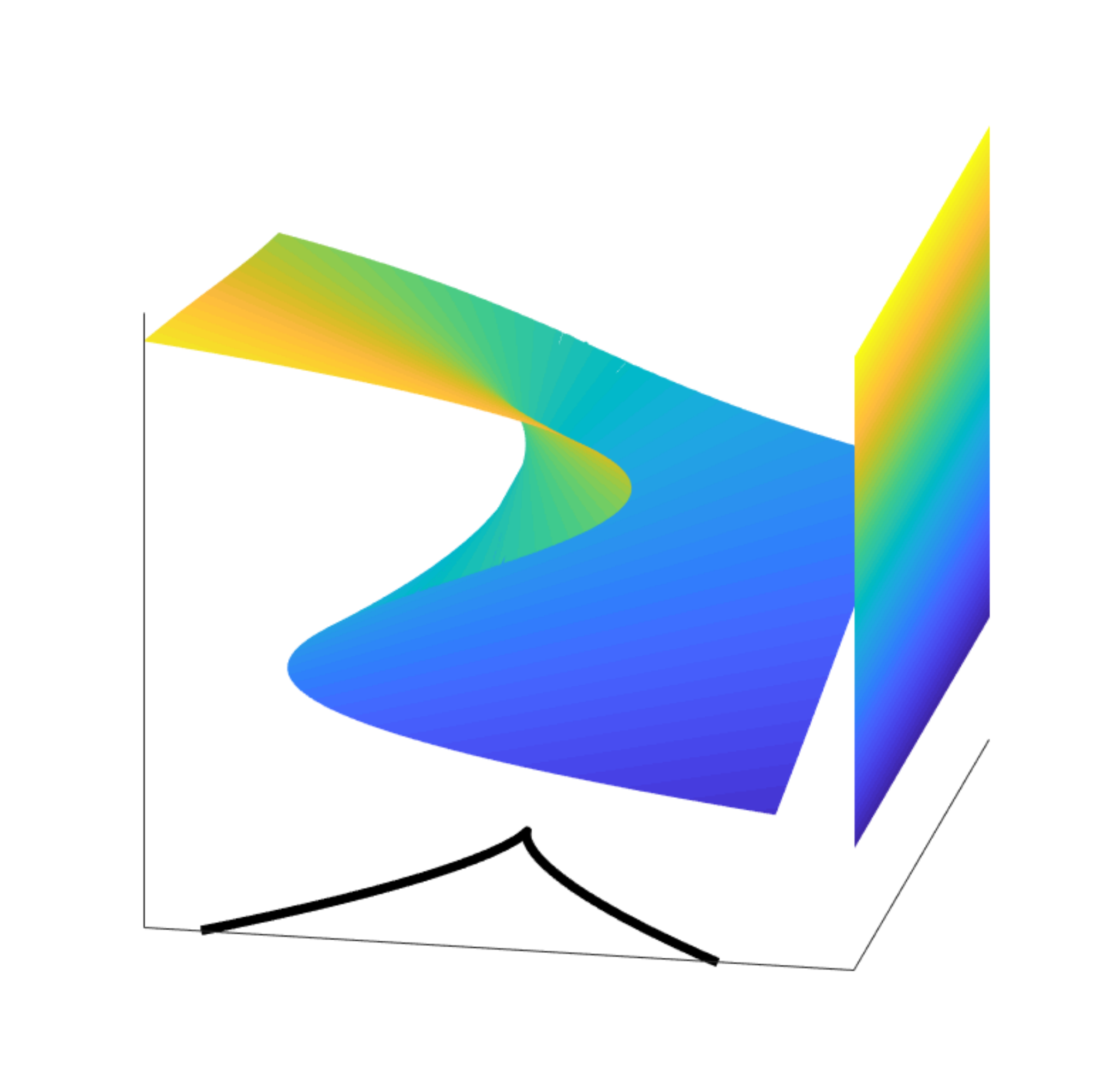}};
\node[inner sep=0pt] (caption_hist) at (0,9.5)
    {Recorded histogram on $(\parameterOne,\parameterTwo)$};   
\node[inner sep=0pt] (mu) at (0,7)
    {\includegraphics[width=.35\textwidth]{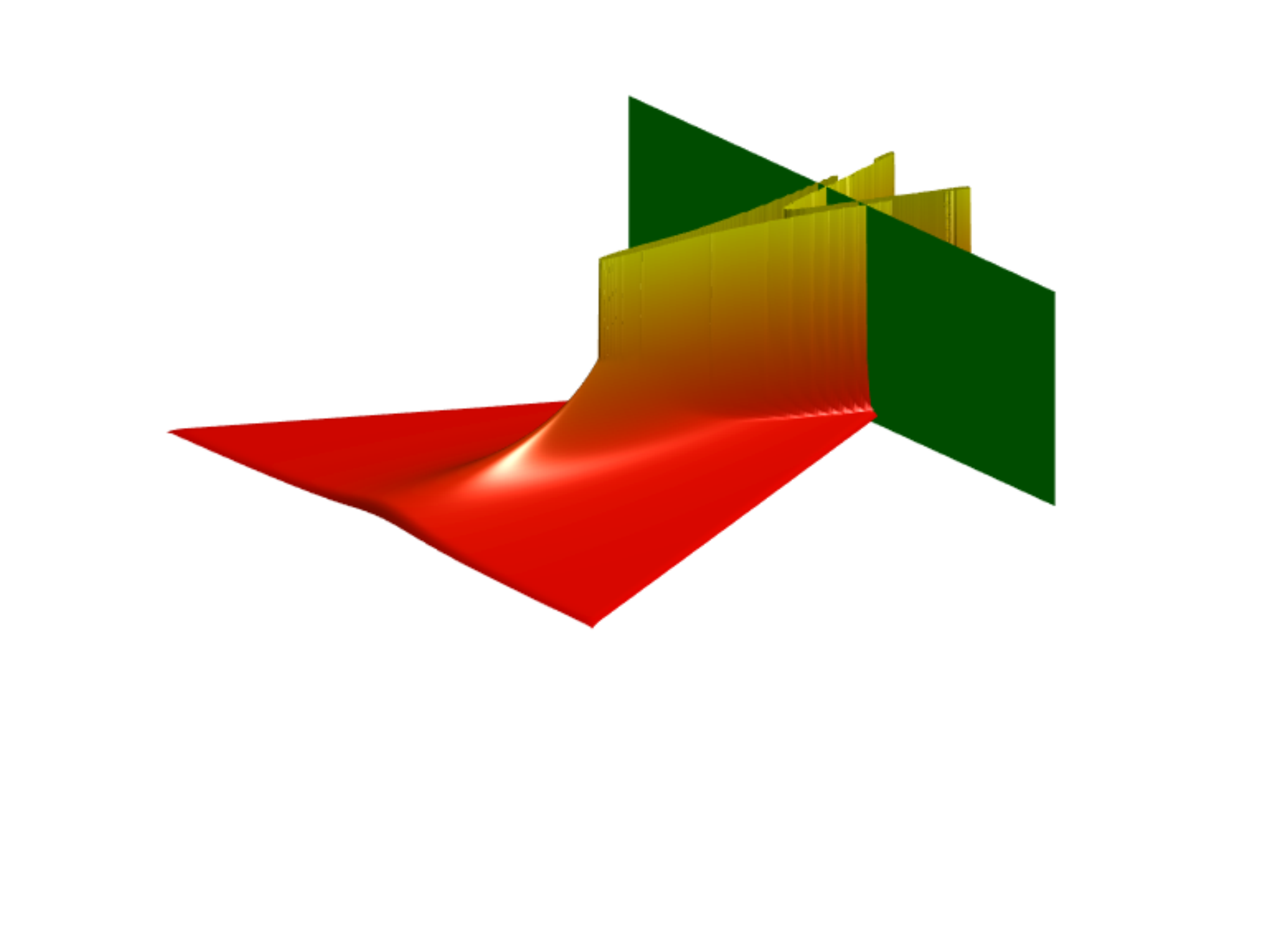}};
\node[inner sep=0pt] (mu) at (-0.3,5.75)
    {\includegraphics[width=.2\textwidth]{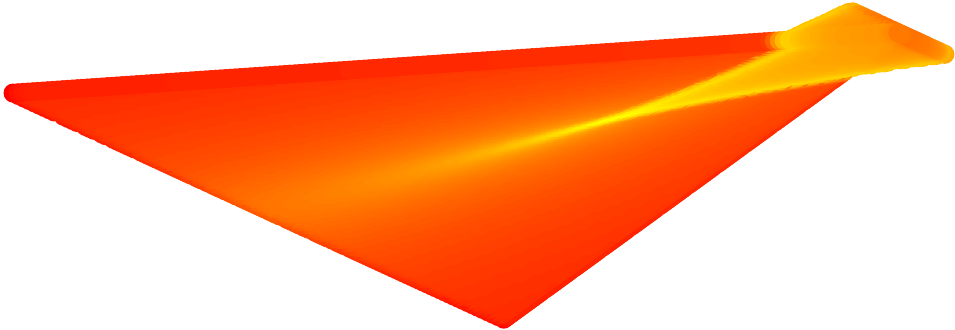}}; 
\node[inner sep=0pt] (mu_south) at (9,4.6){};
\node[inner sep=0pt] (mu_east) at (3,7){};
\node[inner sep=0pt] (OMT_cusp) at (9,0)
    {\includegraphics[width=.4\textwidth]{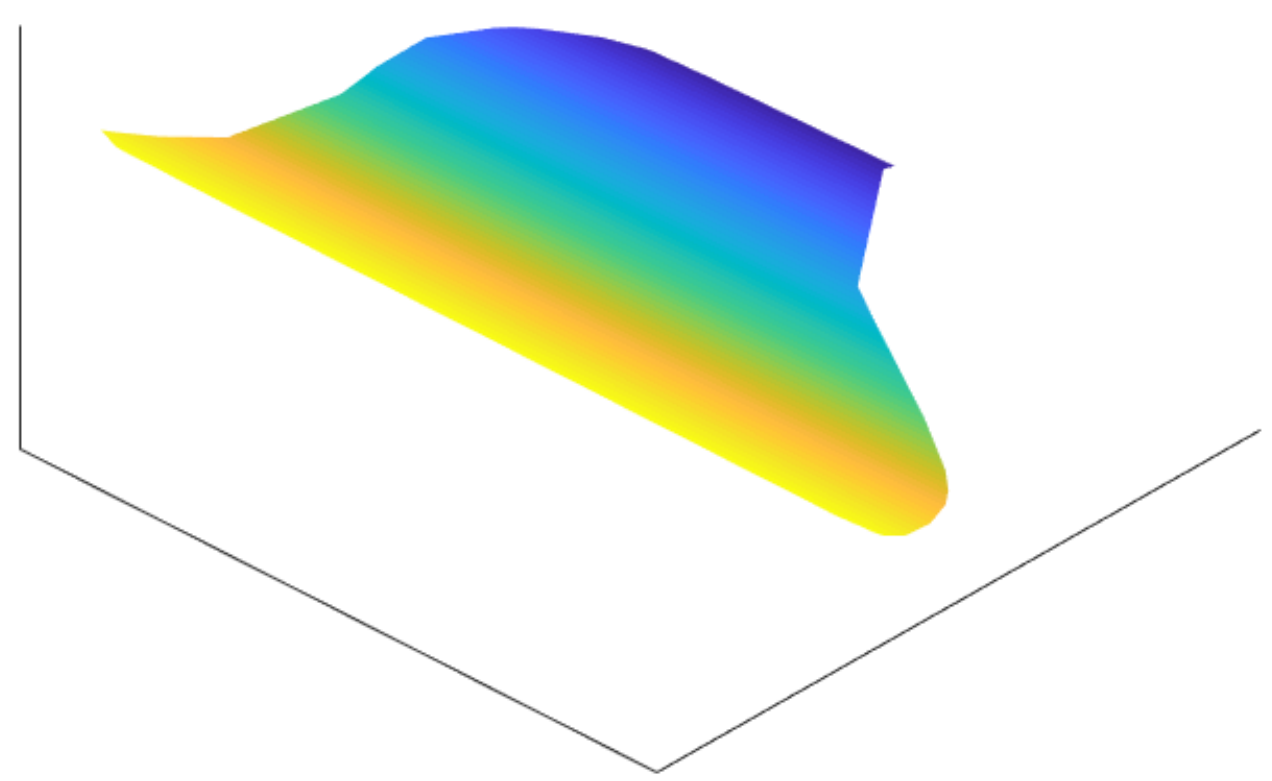}};
\node[inner sep=0pt] (cusp_2D_unfold) at (0.4,-0.5)
    {\includegraphics[width=.25\textwidth]{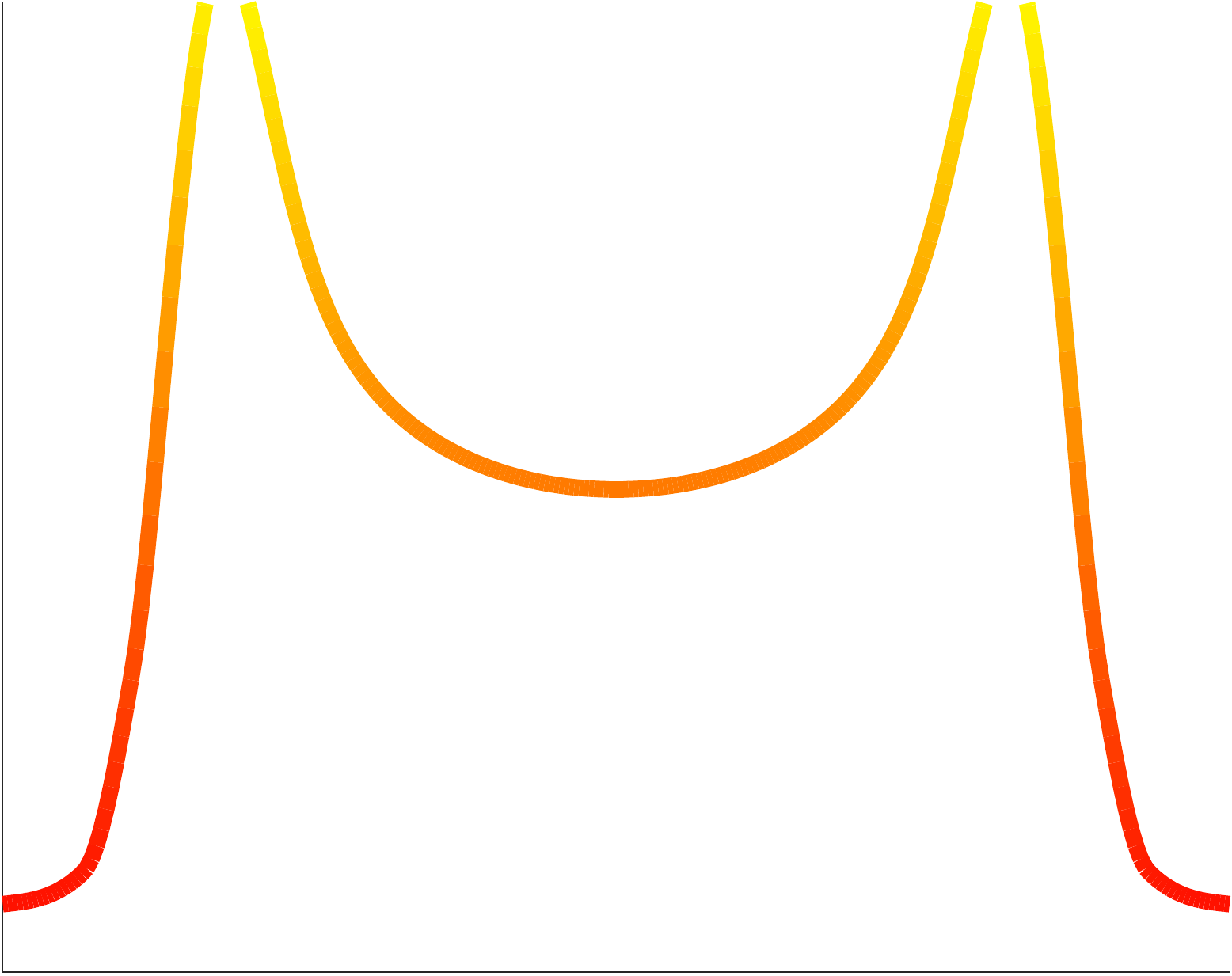}};
\node[inner sep=0pt] (unfold_north) at (-0.2,1.5){};
\node[inner sep=0pt] (PCA1) at (11,-1.5){$\operatorname{PCA}_1$};
\node[inner sep=0pt] (PCA2) at (6.8,-1){$\operatorname{PCA}_2$};
\node[inner sep=0pt] (PCA3) at (5.4,0.7){$\operatorname{PCA}_3$};
\node[inner sep=0pt] (beta2) at (0.2,-2.5){$\parameterTwo$};
\node[inner sep=0pt] (beta2) at (-1.8,-0.3){\rotatebox{90}{density}};
\node[inner sep=0pt] (cusp1) at (11,5.5){$\parameterOne$};
\node[inner sep=0pt] (cusp2) at (8.8,5){$\parameterTwo$};
\node[inner sep=0pt] (cusp3) at (7,6.6){$\systemState$};
\node[inner sep=0pt] (density1) at (-1.1,5.2){$\parameterTwo$};
\node[inner sep=0pt] (density1) at (1.3,5.4){$\parameterOne$};
\node[inner sep=0pt] (beta2) at (-2.5,7.3){\rotatebox{90}{density}};
\node[inner sep=0pt,text width=3.5cm] (caption) at (0,3.4){density in green slice at $(\parameterOne=2/3,\parameterTwo)$};
\draw[-] (-2.2,8.5) to (-2.2,6); 
\draw[-] (-2.2,6) to (-0.1,5);
\draw[-] (-0.1,5) to (1.8,6);
\draw[->,ultra thick,bend angle=0, bend left] (cusp.west) to (mu_east.east);
\draw[->, ultra thick,bend angle=0, bend left] (2.5,5)  to node[midway,fill=white,text width=4cm]{parametrization in $([\parameterOne]_{\delayB},[\parameterTwo]_{\delayA},[\parameterTwo]_\delayB)$} (OMT_cusp.north west);
\draw[->, ultra thick,bend angle=0, bend left] (1.8,6.5) to (1.8,1.5);
\draw[<->,dashed,thick,bend angle=0, bend left] (cusp.south) to node[midway,fill=white]{diffeomorphic} (OMT_cusp.north);
\end{tikzpicture}
\caption{
\emph{Top right}: The cusp surface, together with its parametrization in $(\systemState,\parameterOne)$ and the cusp (black curve), shown as the projection of the folds of the surface on the plane of the two parameters. We treat this surface as the intrinsic, unknown manifold. We observe, for each randomly chosen initial condition $(\systemState_{\delayA},[\parameterOne]_{\delayA})$, the values $[\parameterOne]_{\delayB}$, $[\parameterTwo]_\delayA$ and $[\parameterTwo]_{\delayB}$, as the observation process moves in $\parameterOne$-direction.
\emph{Top left:}
Density obtained by considering only $([\parameterOne]_\delayA,[\parameterTwo]_\delayA)$ observations; note the one-parameter family of infinities. The same density is projected on the $([\parameterOne]_\delayA,[\parameterTwo]_\delayA)$-plane, where yellow color indicates higher density.
\emph{Bottom left:}
For fixed $\parameterOne=2/3$, we observe the density over $\parameterTwo$, i.e.\ the density in the green slice shown in the top left panel; notice the two infinities.
\emph{Bottom right:}
Using principal components (PCA) of $([\parameterOne]_{\delayB}, [\parameterTwo]_\delayA,[\parameterTwo]_{\delayB})$, we can reconstruct the surface. A parametrization of this surface can be transported to the original parametrization in $(\systemState,\parameterOne)$. 
\emph{Coloring:}
The color in the left two plots indicates increasing density (from red to yellow), while the color in the right two plots is with respect to increasing $\systemState$-values (from blue to yellow).}
\label{fig:cusp_density_3D}
\end{figure}

Consider the cusp surface (\Cref{fig:cusp_density_3D}, top right) for $\systemState,\parameterOne \in [-1,1]$ and $\parameterTwo(\systemState,\parameterOne) = \systemState^3 -\parameterOne \, \systemState$.
Sampling uniformly on the $(\systemState,\parameterOne)$-square and observing the distribution of points in $(\parameterOne,\parameterTwo)$ gives rise to a density with a one-parameter family of discontinuities at which the density approaches infinity (see \Cref{fig:cusp_density_3D}, left), including the cusp point $(\parameterOne,\parameterTwo)=(0,0)$.

To ``unfold" this two-dimensional singular density, we assume access to data from an observation process on the $(\systemState,\parameterOne)$ plane, starting from randomly chosen initial conditions and moving in the positive $\parameterOne$-direction. The observations in this example are $([\parameterOne]_\delayB,[\parameterTwo]_\delayA,[\parameterTwo]_{\delayB})$; the $\systemState$ coordinate is not recorded.
\Cref{fig:cusp_density_3D} illustrates the embedding of the $(\systemState,\parameterOne)$-plane into the space of the three principal components of the collection of these delayed observations (bottom right panel).
A two-dimensional parametrization of this
reconstructed surface (e.g.\ through DMAP) can be used to transport the points to the original parametrization in  $(\systemState,\parameterOne)$.

\begin{figure}[t!]
\centering
\begin{tikzpicture}
\node[inner sep=0pt] (caption_cusp) at (9,9.5)
    {Unknown parametrization, unknown manifold};
\node[inner sep=0pt] (cusp) at (9,7)
   {\includegraphics[width=.3\textwidth]{cusp_paramet_proj.pdf}};
\node[inner sep=0pt] (caption_hist) at (0,9.5)
    {Recorded histogram on $\parameterTwo$};   
\node[inner sep=0pt] (mu) at (0,7.3)
    {\includegraphics[width=.3\textwidth]{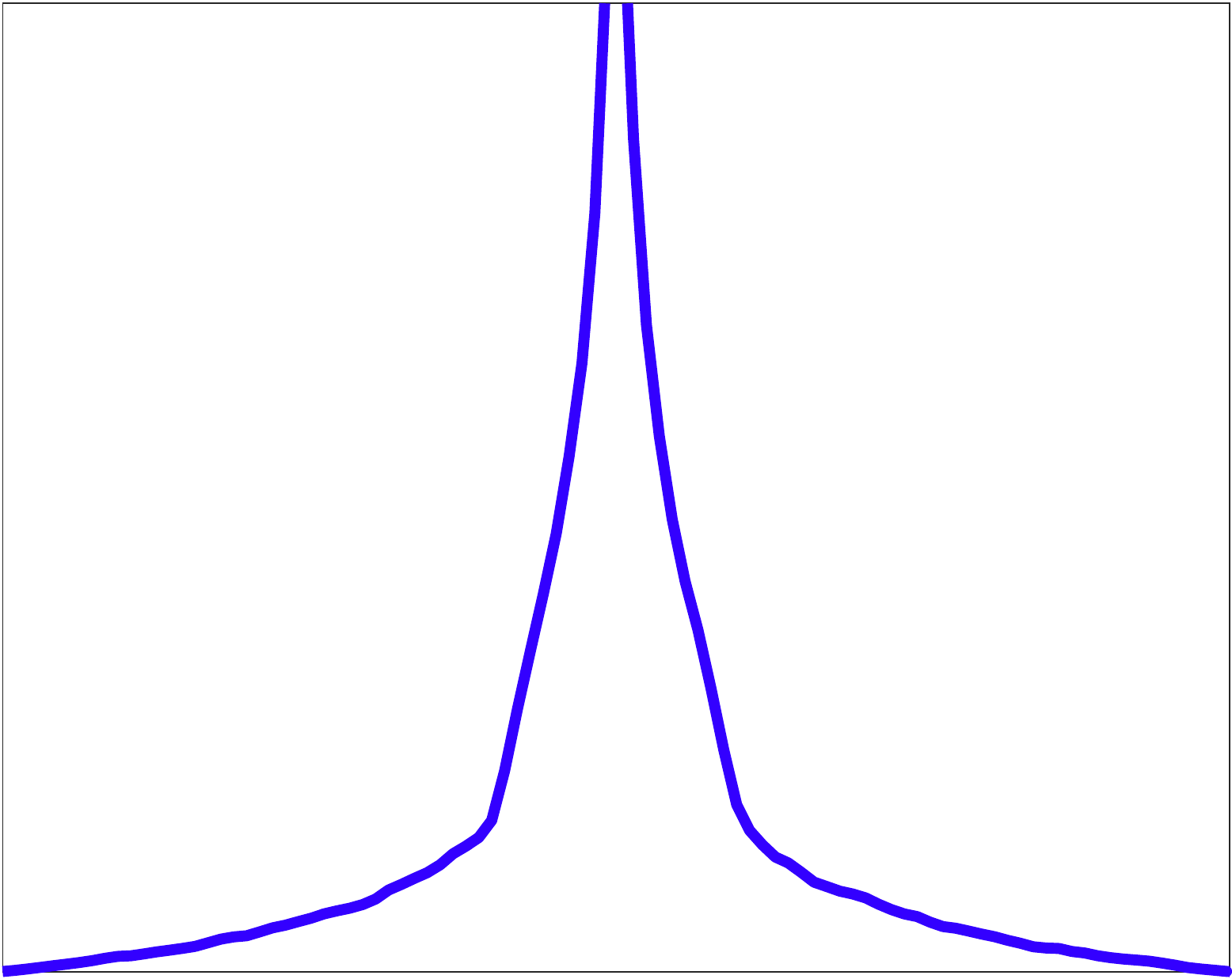}};
\node[inner sep=0pt] (mu_south) at (9,4.6){};
\node[inner sep=0pt] (mu_east) at (3,7){};
\node[inner sep=0pt] (OMT_cusp) at (9,-0.3)
    {\includegraphics[width=.4\textwidth]{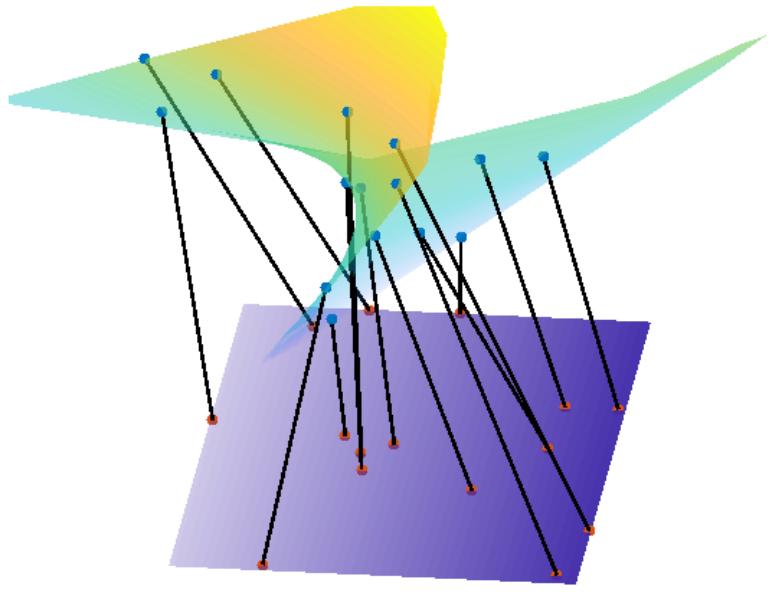}};
\node[inner sep=0pt] (cusp_2D_unfold) at (0.3,-0.5)
    {\includegraphics[width=.3\textwidth]{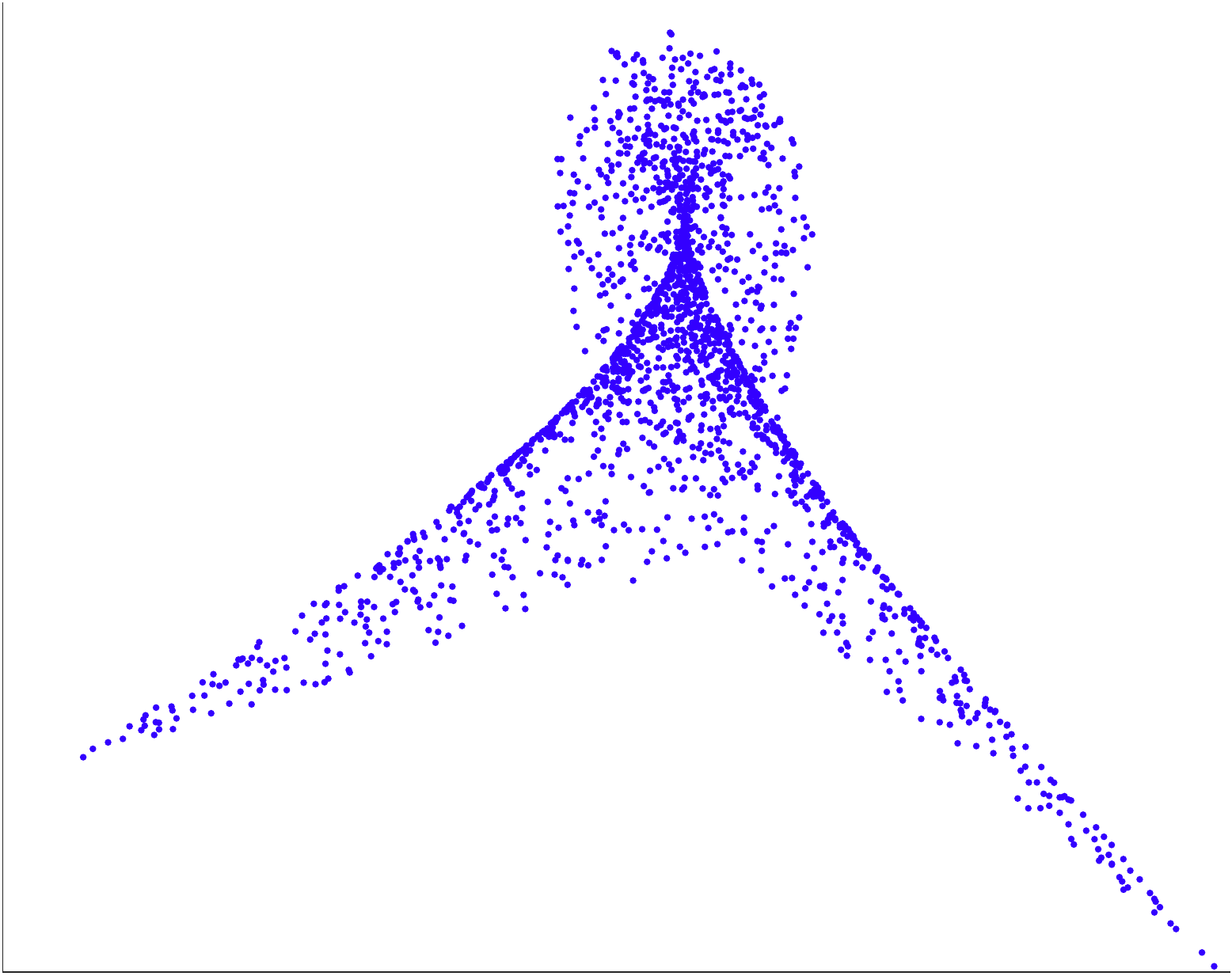}};
\node[inner sep=0pt] (unfold_north) at (0,1.5){};
\node[inner sep=0pt] (x) at (9,-2.8){$\systemState$};
\node[inner sep=0pt] (lambda) at (11.2,-1.5){$\parameterOne$};
\node[inner sep=0pt] (cusp1) at (11,5.5){$\parameterOne$};
\node[inner sep=0pt] (cusp2) at (8.8,5){$\parameterTwo$};
\node[inner sep=0pt] (cusp3) at (7,6.6){$\systemState$};
\node[inner sep=0pt] (cusp2) at (0,5.2){$\parameterTwo$};
\node[inner sep=0pt] (cusp3) at (-2.5,7.5){\rotatebox{90}{density}};
\node[inner sep=0pt] (cusp3) at (-2.3,-0.2){\rotatebox{90}{$[\parameterTwo]_{n+1}$}};
\node[inner sep=0pt] (cusp3) at (0.3,-2.7){$[\parameterTwo]_n$};
\draw[->,ultra thick,bend angle=0, bend left] (cusp.west) to (mu_east.east);
\draw[->, ultra thick,bend angle=0, bend left] (mu.south east)++(-0.25,-0.25) to node[midway,fill=white,text width=4cm]{parametrization in $([\parameterTwo]_{\delayA},[\parameterTwo]_{\delayB},[\parameterTwo]_\delayC)$} (OMT_cusp.north west);
\draw[->, ultra thick,bend angle=0, bend left] (mu.south)++(0,-0.5) to node[midway,fill=white,text width=3.5cm]{no embedding in $([\parameterTwo]_{\delayA},[\parameterTwo]_\delayB)$} (unfold_north.north);
\draw[<->,dashed,thick,bend angle=0, bend left] (cusp.south) to node[midway,fill=white]{diffeomorphic} (OMT_cusp.north);
\end{tikzpicture}
\caption{
\emph{Top right}: The same surface as in \Cref{fig:cusp_density_3D}. Again, we treat this surface as the intrinsic, unknown manifold. We observe consecutive values of $\parameterTwo$ starting from uniformly distributed points in $(\systemState,\parameterOne)$ and moving in the $\systemState$-direction.
\emph{Top left:}
The recorded histogram of $[\parameterTwo]_n$ values.
\emph{Bottom left:}
From two consecutive $\parameterTwo$ values, we obtain what clearly are overlaps of projecting a surface. This shows that this two-dimensional embedding is not enough to unfold the singularity (in contrast to the example we considered in \Cref{sec:hook}).
\emph{Bottom right:}
With three consecutive $\parameterTwo$ values we can
reconstruct the surface using a delay-embedding. A two-dimensional parametrization of the surface can be used to transport the points to the original parametrization (illustrated schematically through black lines from the surface to $(\systemState,\parameterOne)$). 
This provides a meaningful way to ``transport" the one-dimensional marginal density (top left) to a two-dimensional joint density (bottom right), a task that in principle is not well-defined.}
\label{fig:OMT_cusp}
\end{figure}

\subsection{Converting one-dimensional marginal distributions to joint distributions}
\label{sec:marginal}

Given the marginal density in $\parameterTwo$ (\Cref{fig:OMT_cusp}, top left), we cannot transport it to the uniform density in $(\systemState,\parameterOne)$, as the respective dimensions disagree.
Using history from an observation process starting at randomly selected initial points in $(\systemState,\parameterOne)$ and taking two steps in the $\systemState$-direction, we could attempt (in the spirit of \Cref{sec:hook}) to ``unfold" the singularity at $([\parameterTwo]_n=0)$ using the delays $([\parameterTwo]_n,[\parameterTwo]_{\delayB})$. 
In contrast to the example of \Cref{sec:hook}, this clearly does not lead to a curve in $\RR^2$, but visibly to overlaps of a surface in the projection, see \Cref{fig:OMT_cusp} (bottom left). This indicates that the intrinsic dimension is not one, but two, and that we need a third coordinate to properly embed the underlying manifold. This visual assertion can be also quantified by applying dimension estimation algorithms to the ensemble of possible observation histories. 

In delay-coordinates $([\parameterTwo]_{\delayA},[\parameterTwo]_{\delayB},[\parameterTwo]_{\delayC})$ we are able to unfold the scalar $[\parameterTwo]_n$ observations (\Cref{fig:OMT_cusp} top left), to a two-dimensional surface (bottom right, top of the figure, showing the data set embedded into the three principal components of the observation delay-coordinates). The resulting surface is diffeomorphic to the original cusp surface. As in the previous section, a two-dimensional parametrization of this
reconstructed surface could be used to transport the points to the original $(\systemState,\parameterOne)$-parametrization (illustrated schematically through black lines from the surface to the $(\systemState,\parameterOne)$ square in \Cref{fig:OMT_cusp}, bottom right).

\subsection{Parametrizing the embedded surface}
One could attempt to map the two-dimen\-sional surface reconstructed from ``unfolding" the one-dimensional marginal distribution of $\parameterTwo$ values in \Cref{fig:OMT_cusp} (as detailed in \Cref{sec:marginal}) to the two original coordinates $(\systemState,\parameterOne)$ through Wasserstein optimal transport.
Alternatively, we can also recover the original parametrization $(\systemState,\parameterOne)$ using a non-linear manifold learning technique, Diffusion Maps (DMAP)~\cite{coifman-2006} (see \Cref{sec:appendix dmaps}). To do this, we must employ the Mahalanobis distance \cite{coifman-2006,singer-2008}, see \Cref{fig:cusp_mahala}, which, however, requires more information than just densities or even ensembles of points on each two-dimensional surface: 
It needs estimates of the covariance matrices (the Jacobian of the mapping between the two manifolds) at each point. We discuss the difference between the Wasserstein and Mahalanobis frameworks for transport in \Cref{sec:mahala_wasser}.

\begin{figure}[t!]
\centering
\begin{tikzpicture}
\node[inner sep=0pt] (manifold) at (0,0)
    {\includegraphics[scale=0.7]{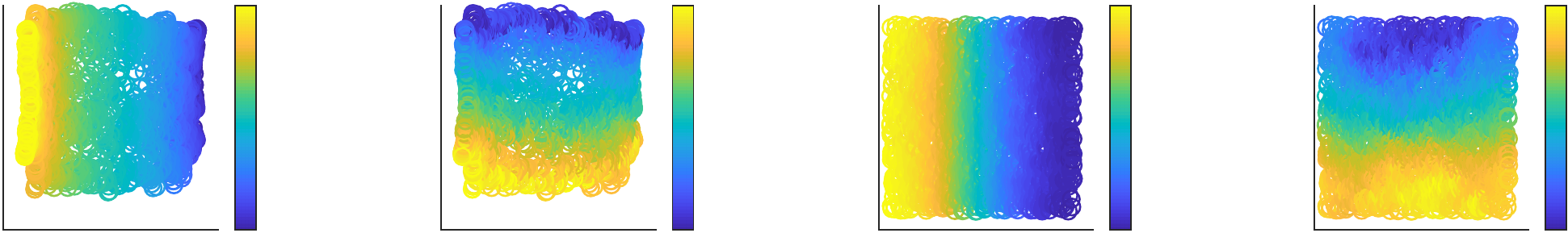}};
\node[inner sep=0pt] (phi10) at (-5.8,-1.3){$\phi_1^M$};
\node[inner sep=0pt] (phi20) at (-7.3,0){\rotatebox{90}{$\phi_3^M$}};
\node[inner sep=0pt] (phi21) at (-3.3,0){\rotatebox{90}{$\phi_3^M$}};
\node[inner sep=0pt] (phi11) at (-1.8,-1.3){$\phi_1^M$};
\node[inner sep=0pt] (x0) at (1.8,-1.3){$\systemState$};
\node[inner sep=0pt] (beta11) at (0.5,0){\rotatebox{90}{$\parameterOne$}};
\node[inner sep=0pt] (beta12) at (4.4,0){\rotatebox{90}{$\parameterOne$}};
\node[inner sep=0pt] (x1) at (5.8,-1.3){$\systemState$};
\node[inner sep=0pt] (phi20) at (-4.3,0){\rotatebox{90}{$\systemState$}};
\node[inner sep=0pt] (phi20) at (-0.5,0){\rotatebox{90}{$\parameterOne$}};
\node[inner sep=0pt] (phi20) at (7.2,0){\rotatebox{90}{$\phi_3^M$}};
\node[inner sep=0pt] (phi20) at (3.4,0){\rotatebox{90}{$\phi_1^M$}};
\end{tikzpicture}
\caption{Application of DMAP with Mahalanobis distance to the PCA parametrization of the unfolded surface of \Cref{fig:OMT_cusp} (bottom right). This recovers an embedding of the original uniformly distributed coordinates $(\systemState,\parameterOne)$; here the correspondence is clearly visible.
\emph{First two plots}: Mahalanobis-DMAP coordinates $(\phi^M_1,\phi^M_3)$ colored by the original, uniformly distributed coordinates $(\systemState,\parameterOne)$.
\emph{Second two plots}: Coordinates $(\systemState,\parameterOne)$ are colored by the Mahalanobis-DMAP coordinates $(\phi^M_1,\phi^M_3)$.}
\label{fig:cusp_mahala}
\end{figure}
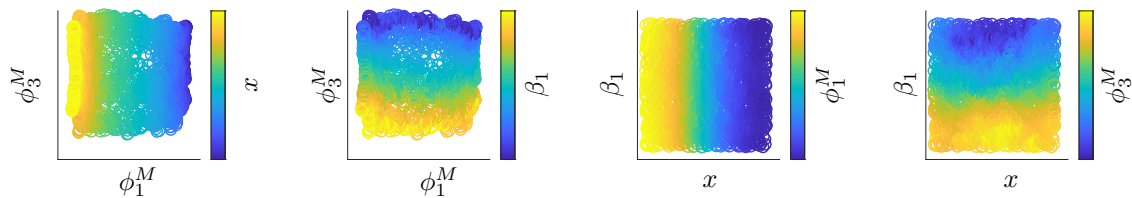

\section{Recovering parametrizations up to isometries}\label{sec:mahala_wasser}

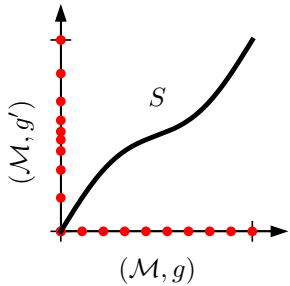
\begin{figure}[t!]
\centering
\begin{tikzpicture}
\node[inner sep=0pt] (manifold) at (0,0)
    {\includegraphics[width=.3\textwidth]{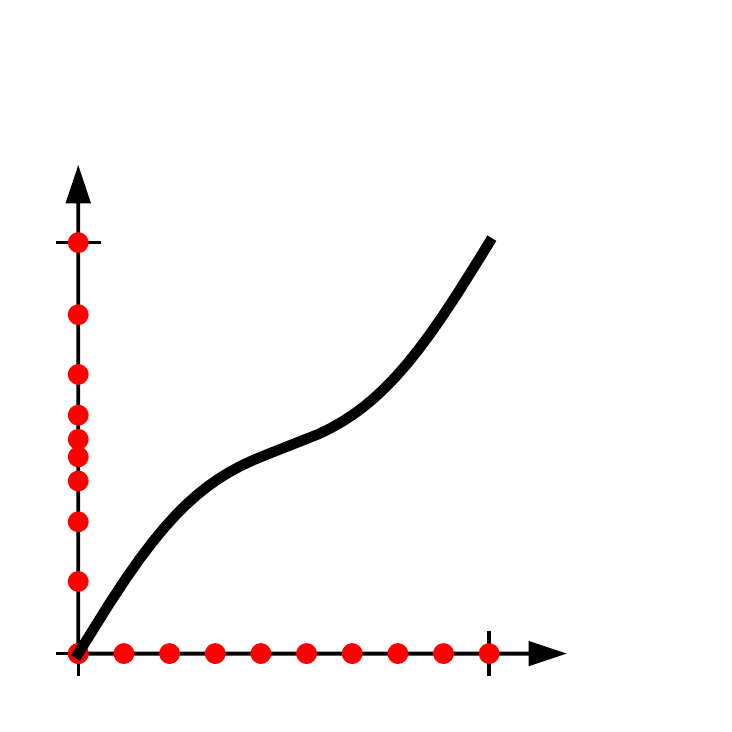}};
\node[inner sep=0pt] (cusp1) at (-0.5,-2.3){$(\M,g)$};
\node[inner sep=0pt] (cusp1) at (-2.3,-0.6){\rotatebox{90}{$(\M,g')$}};
\node[inner sep=0pt] (cusp1) at (-0.5,0){$S$};
\end{tikzpicture}
\caption{A manifold $(\M,g)$ (in this example $\M=[0,1]$ and $g$ is the Euclidean metric) is mapped to itself through an invertible function $S:\M\to\M$. This map induces a new metric $g'=S_*g$ on $\M$. The two axes indicate how measures, illustrated as red point distributions, are mapped by $S$: The uniform density of $\vol_g$ (the measure on $\M$ induced by the Lebesgue measure on $\RR$, horizontal axis) is mapped by $S$ to a uniform density w.r.t $\vol_{g'}$ on $\M$ (vertical axis).}
\label{fig:manifold_map_S}
\end{figure}

In the previous sections, we illustrated the recovery of manifolds from (individually) non-invertible observations.
In this section, we discuss the conceptual similarities and differences between the reconstruction of geometry through Mahalanobis-Diffusion Maps and the construction of transport maps optimal in the Wasserstein sense. We simplify the presentation by assuming that the observations are now already embeddings of the manifold (i.e., invertible on their image).
\Cref{fig:manifold_map_S} illustrates the concept of maps changing metrics and measures.

Let $S$ be a smooth, invertible function, mapping a manifold $\M$ to itself. This map changes the metric $g$ into the metric $S_{\ast}g$. The density of points on $\M$ also changes through $S$, indicating that measures are also transformed.
This section discusses how the change in the metric and the resulting change in the measures relate to each other, and how the information about either metrics or measures can be used to recover different aspects of the map $S$.

Given two measures $\mu$ and $\nu$ on a Riemannian manifold $(\M,g)$, absolutely continuous with respect to its volume form $\vol_g$ (defined in \Cref{sec:optimal transport on riemannian manifolds}), we endeavor to reconstruct a fixed, but unknown, invertible, {\em measure-preserving} map $S:\M\to \M$ such that
$
    \nu(A)=\mu(S^{-1}(A)),\ A\subset \M,
$
which we write as $S_\sharp\mu=\nu$. As discussed in \Cref{sec:notation}, there may exist many such maps for fixed $\mu$ and $\nu$.
When trying to recover the original function $S$, we may choose an \emph{optimality} criterion to at least select a unique map---for example, we could use the map that is minimizing the Wasserstein cost. In this section, we consider measures that are induced by metrics, which can lead to a notion of optimality via \emph{metric preservation}.

By solving the Wasserstein optimal transport problem between $\mu$ and $\nu$ on $\M$, we can reconstruct the map $S$ only up to maps $U:\M\to \M$ with $U_\sharp\mu=\mu$. This is due to the \emph{polar factorization}~(\cite{brenier-1991,mccann-2001} and \Cref{thm:exist_OT_RM} in the Appendix):
Given a measure $\mu \ll \vol_g$, a Borel map $S: \M \to \M$ can be written uniquely ($\mu$-a.e) as $S=T\circ U$, where $T$ is the optimal transport between $\mu$ and $S_{\sharp}\mu$ and $U$ is measure-preserving w.r.t. $\mu$, i.e.\ $U_{\sharp}\mu = \mu$.

Instead of optimizing the Wasserstein cost and reconstructing $S$ up to measure-preserving maps $U$, we employ here {\em additional information} from our observation processes to obtain a map with far less ambiguity.
To this end, we use the construction of {\em metric-preserving} maps described in \cite{singer-2008,berry-2016}:
Given a metric $g$ and its push-forward metric $S_{\ast}g$ by a diffeomorphism $S: \M \to \M$, the map $S$ can be reconstructed 
up to a linear, orthogonal map. 
The reconstruction can even be done in a data-driven way, 
employing diffusion map (DMAP) embeddings \cite{coifman-2006} and a Mahalanobis distance~\cite{singer-2008}.
The reconstruction up to an orthogonal map is justified through the following argument: By choosing the metric $S_{\ast}g$ based on the metric $g$, $S$ is an isometry. 
Laplace-Beltrami operators of isometric manifolds have the same eigenvalues, and eigenfunctions associated to the same eigenvalue are related by an orthogonal map \cite{berry-2016,rosenberg-1997}.
Therefore, an isometry $S$ between the base manifolds $(\M,g)$ and $(\M,S_*g)$ turns into an orthogonal map in eigenfunction coordinates (Mahalanobis-DMAP coordinates) of the manifolds, which can be computed easily \cite{berry-2016}. Similar to \cite{berry-2016}, we summarize this in a commutative diagram (\Cref{fig:dmap_orthogonal}).

The metric $S_{\ast}g$ can be computed from $g$ by
$
    (S_{\ast}g)_y(\xi,\eta) = g_{S^{-1}(y)}(J(y)\xi,J(y)\eta),
$
where $\xi,\eta \in T_y\M$ and $J(y)$ denotes the Jacobian of $S^{-1}$ at $y\in \N$. If $\N$ is embedded in Euclidean space, to use the push-forward metric $S_{\ast}g$ instead of the induced Euclidean metric on $\N$ for DMAP embeddings, a special kernel
can be employed (see \eqref{eq:Mahala_kernel} in \Cref{sec:appendix dmaps}, and \cite{singer-2008}). The kernel requires estimations of the Jacobian matrices of $S$ at every point, which can be obtained by computing local covariance matrices as described in \cite{singer-2008}. The induced metric $S_{\ast}g$ computed via covariance matrices is also referred to as the \emph{Mahalanobis distance}, see \cite{dsilva-2016}.

Even if $S$ is not an isometry, but only measure-preserving between the manifolds $(\M,g)$ and $(\M,g')$, the map between their diffusion map embeddings is still measure-preserving (see \cref{fig:dmap_measure}): An embedding with $\ell$ eigenfunctions is an isometry between the original manifold $(\M,g)$ and $\Phi_g(\M) \subset \RR^\ell$ with the metric induced by the Euclidean metric, up to the truncation error caused by using a finite number of eigenfunctions \cite{portegies-2016}. Therefore, measures are also preserved up to this accuracy.

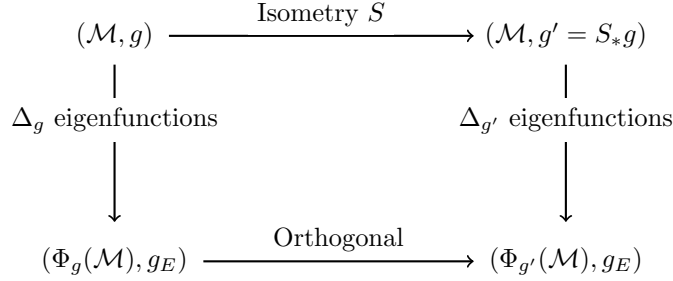
\begin{figure}[t!]
\centering
\begin{tikzpicture}
\node[inner sep=0pt] (X) at (0,3)
    {$(\M,g)$};
\node[inner sep=0pt] (Y) at (6,3)
    {$(\M,g'=S_{\ast}g)$};
\node[inner sep=0pt] (Y_aux) at (4.8,3)
    {};    
\node[inner sep=0pt] (PhiX) at (0,0)
    {$(\Phi_g(\M),g_E)$};
\node[inner sep=0pt] (PhiX_aux) at (0,0.5)
    {};
\node[inner sep=0pt] (PhiY) at (6,0)
    {$(\Phi_{g'}(\M),g_E)$};
\node[inner sep=0pt] (PhiY_aux1) at (6,0.5)
    {};
\node[inner sep=0pt] (PhiY_aux2) at (4.8,0)
    {};
\draw[->,thick,bend angle=0, bend left] (X.east)++(0.2,0) to node[above,midway,fill=white]{Isometry $S$}
(Y_aux.west);
\draw[->,thick] (X.south)++(0,-0.2) to node[above,midway,fill=white]{$\Delta_{g}$ eigenfunctions}  (PhiX_aux.north);
\draw[->,thick,bend angle=0, bend right] (PhiX.east)++(0.2,0) to  node[above,midway]{Orthogonal} (PhiY_aux2.west);
\draw[->,thick,bend angle=0, bend right] (Y.south)++(0,-0.2) to  node[above,midway,fill=white]{$\Delta_{g'}$ eigenfunctions} (PhiY_aux1.north);
\end{tikzpicture}
\caption{A commutative diagram showing how isometries are represented as orthogonal mappings in Laplacian eigenfunction coordinates, see \cite{berry-2016}. Here $\Delta_g$ denotes the Laplacian with respect to the metric $g$, $\Phi_g$ is the embedding via Laplacian eigenfunctions (cut-off at $\ell$), and $g_E$ denotes the metric on the embedded manifold $\Phi_g({\M})$ induced by the Euclidean metric in $\RR^{\ell}$. Laplacian eigenvalues are the same for isometric manifolds, while eigenfunctions with respect to the same eigenvalue are related by an orthogonal map \cite{rosenberg-1997}. Therefore, the isometry $S$ induces an orthogonal map in Laplacian eigefunction coordinates.}
\label{fig:dmap_orthogonal}
\end{figure}

\begin{figure}[t!]
\centering
\begin{tikzpicture}
\node[inner sep=0pt] (X) at (0,3)
    {$(\M,\vol_g)$};
\node[inner sep=0pt] (Y) at (6.5,3)
    {$(\M,S_{\sharp}\vol_g=\vol_{g'})$};
\node[inner sep=0pt] (Y_aux2) at (6,3)
    {};
\node[inner sep=0pt] (Y_aux) at (4.8,3)
    {};    
\node[inner sep=0pt] (PhiX) at (0,0)
    {$(\Phi_g(\M),\lambda)$};
\node[inner sep=0pt] (PhiX_aux) at (0,0.5)
    {};
\node[inner sep=0pt] (PhiY) at (6,0)
    {$(\Phi_{g'}(\M),\lambda)$};
\node[inner sep=0pt] (PhiY_aux1) at (6,0.5)
    {};
\node[inner sep=0pt] (PhiY_aux2) at (5,0)
    {};
\draw[->,thick,bend angle=0, bend left] (X.east)++(0.2,0) to node[above,midway,fill=white]{measure-preserving $S$}
(Y_aux.west);
\draw[->,thick] (X.south)++(0,-0.2) to node[above,midway,fill=white]{$\Delta_{g}$ eigenfunctions}  (PhiX_aux.north);
\draw[->,thick,bend angle=0, bend right] (PhiX.east)++(0.2,0) to  node[above,midway]{measure-preserving} (PhiY_aux2.west);
\draw[->,thick,bend angle=0, bend right] (Y_aux2.south)++(0,-0.4) to  node[above,midway,fill=white]{$\Delta_{g'}$ eigenfunctions} (PhiY_aux1.north);
\end{tikzpicture}
\caption{Commutative diagram showing how measure-preserving maps induce measure-preserving maps in Laplacian eigenfunction coordinates. Here $\Delta_g$ denotes the Laplacian with respect to the metric $g$, $\Phi_g$ is the embedding via Laplacian eigenfunctions (cut-off at $\ell$), and $\lambda$ denotes the volume form on the embedded manifold $\Phi_g(\M)$ induced by the Lebesgue measure on $\RR^{\ell}$. The metric $g'$ is a metric on $\M$ such that $\vol_{g'}=S_{\sharp}\vol_g$, but $g' \neq S_{\ast}g$ in general (If $g' = S_{\ast}g$ we are in the situation of \Cref{fig:dmap_orthogonal}).}
\label{fig:dmap_measure}
\end{figure}

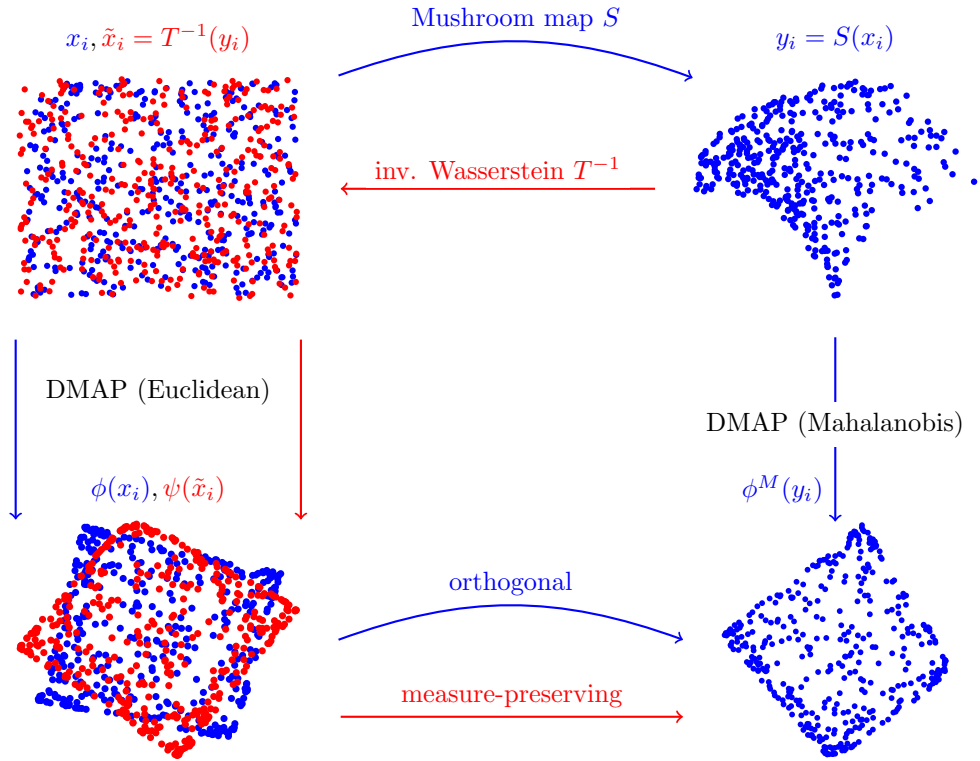
\begin{figure}
\centering
\begin{tikzpicture}
\node[inner sep=0pt] (X) at (0,8)
    {${\color{blue}x_i}, {\color{red}\tilde{x}_i=T^{-1}(y_i)}$};
\node[inner sep=0pt] (BVG1) at (0,6)
    {\includegraphics[width=.25\textwidth]{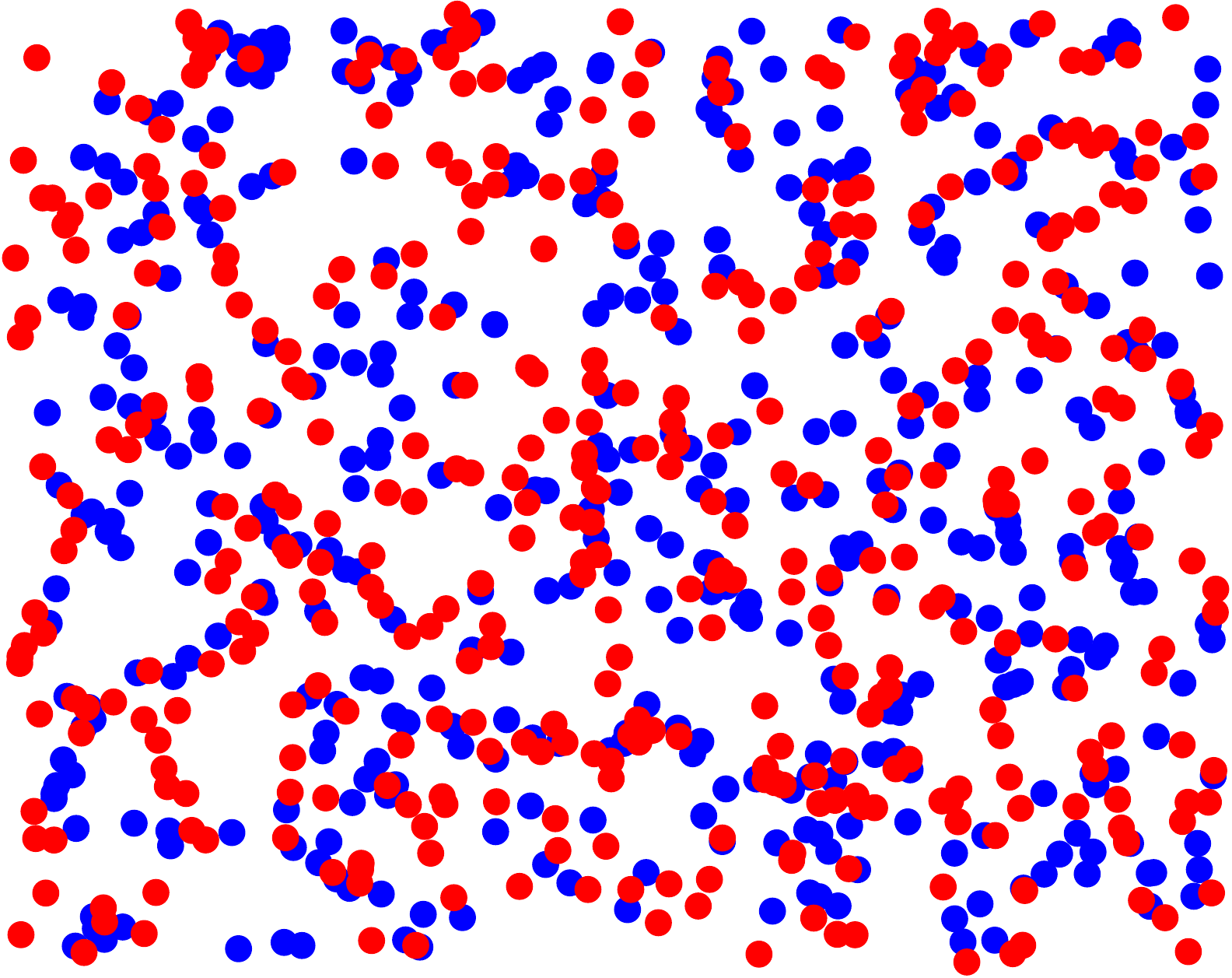}};
\node[inner sep=0pt] (BVG1_aux) at (2.4,6){};
\node[inner sep=0pt] (Y) at (9,8)
    {${\color{blue} y_i=S(x_i)}$};
\node[inner sep=0pt] (BVG2) at (9,6)
    {\includegraphics[width=.25\textwidth]{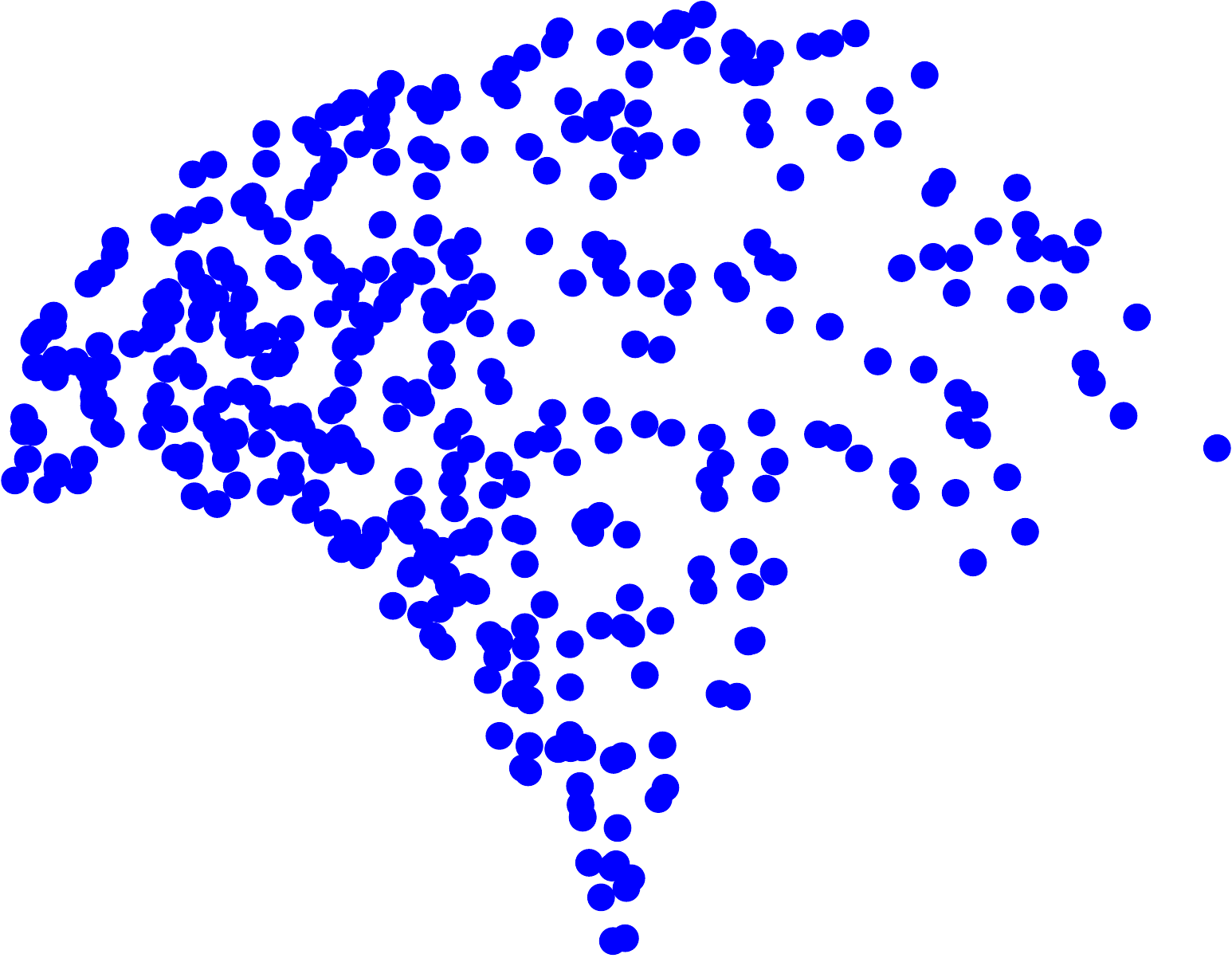}};
\node[inner sep=0pt] (PHI) at (0,2)
    {${\color{blue}\phi(x_i)}, {\color{red} \psi(\tilde{x}_i)}$};
\node[inner sep=0pt] (DMAP_orig) at (0,0)
    {\includegraphics[width=.25\textwidth]{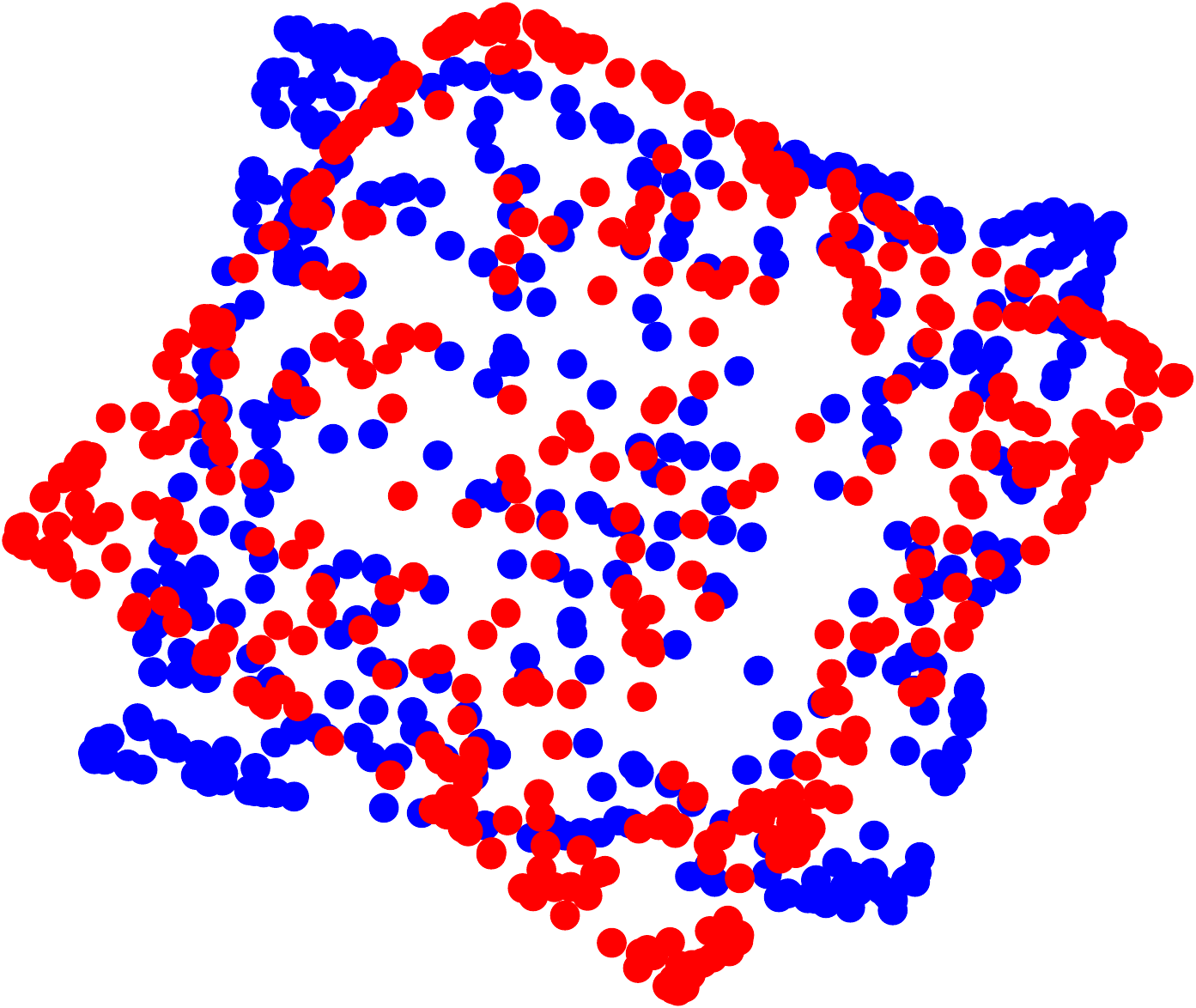}};
\node[inner sep=0pt] (DMAP_orig_aux) at (2.4,-1){};
\node[inner sep=0pt] (PHIM) at (8.3,2)
    {\color{blue} $\phi^M(y_i)$};
\node[inner sep=0pt] (DMAP_rb) at (9,0)
    {\includegraphics[width=.2\textwidth]{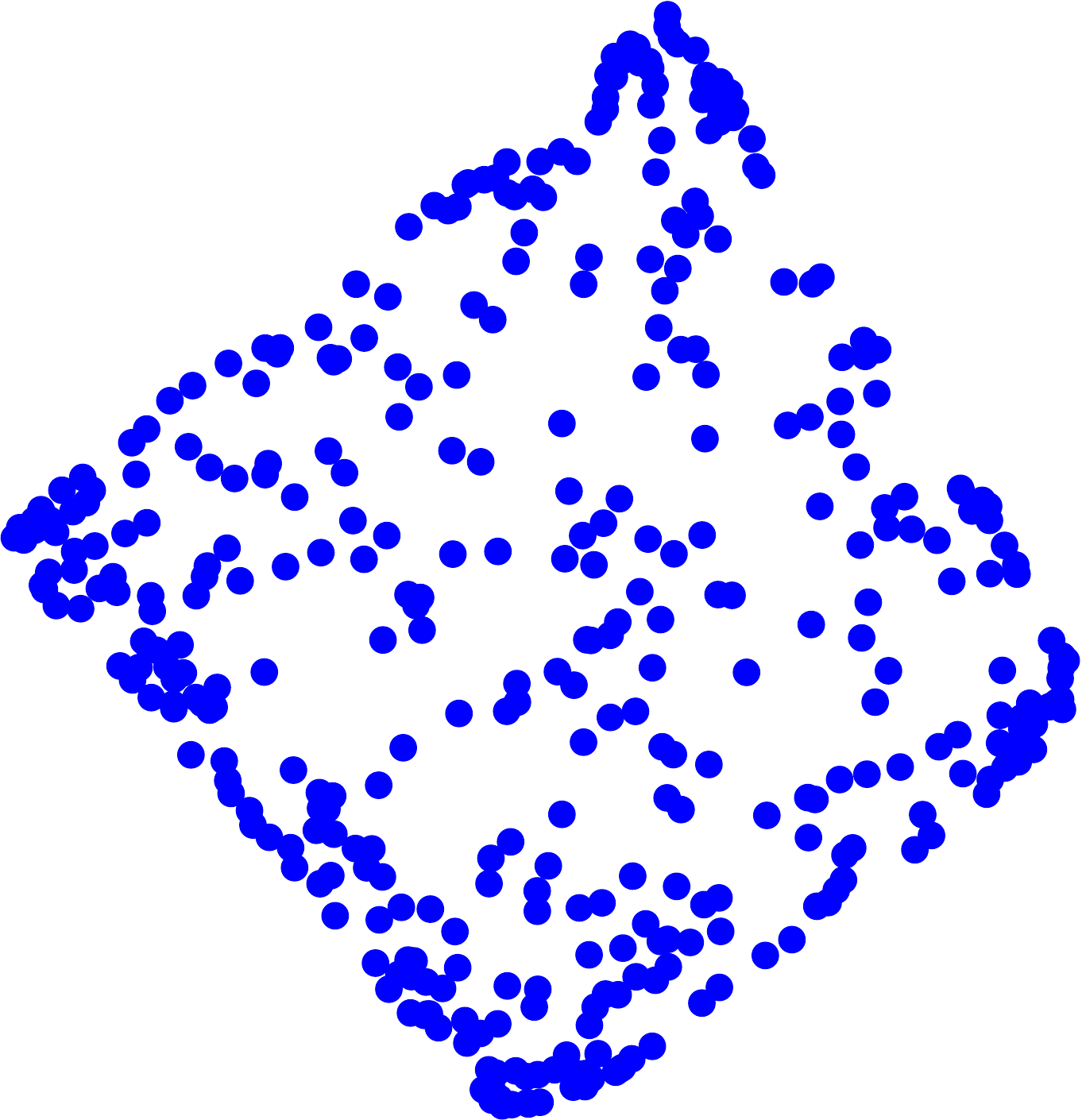}};
\node[inner sep=0pt] (DMAP_rb_aux) at (7,-1){};
\node[inner sep=0pt] (DMAP_rb_aux2) at (7,0){};
\node[inner sep=0pt] (DMAP) at (0,3.3)
    {DMAP (Euclidean)};
\draw[->,thick,bend angle=20, bend left,blue] (BVG1.north east)++(0.5,0) to node[above,midway] {Mushroom map $S$}  (BVG2.north west);
\draw[->,thick,bend angle=0, bend left,red] (BVG2.west)++(-0.5,0) to node[above,midway,fill=white]{inv. Wasserstein $T^{-1}$}
(BVG1_aux.east);
\draw[->,thick,blue] (BVG1.south west)++(0,-0.5) to  (DMAP_orig.north west);
\draw[->,thick,red] (BVG1.south east)++(0,-0.5) to (DMAP_orig.north east);
\draw[->,thick,blue] (BVG2.south)++(0,-0.5) to node[below,pos=0.35, fill=white]{{\color{black}DMAP (Mahalanobis)}} (DMAP_rb.north);
\draw[->,thick,bend angle=20, bend left,blue] (DMAP_orig.east)++(0.5,0) to  node[above,midway]{orthogonal} (DMAP_rb_aux2.west);
\draw[->,thick,bend angle=0, bend right,red] (DMAP_orig_aux.south east) to  node[above,midway]{measure-preserving} (DMAP_rb_aux.south west);
\end{tikzpicture}
\caption{Uniformly sampled points $x_i$ on the unit square (upper left, blue), are mapped with the mushroom map $S$ \eqref{eq:mushroom_map} to $y_i$ (upper right). DMAP with Mahalanobis distance (coming from the Jacobian of $S^{-1}$) is applied to $y_i$, to obtain an embedding $\phi^M(y_i)$ (lower right). The points $y_i$ are also mapped back to the unit square ($\tilde{x}_i$, upper left, red) with the Wasserstein optimal transport (computed from the mushroom-distribution to the uniform distribution on the unit square). DMAPS with Euclidean distance are applied to both the red and the blue square (upper right), resulting in rotated squares of the same color (lower left). The distribution of points in all of the squares in the lower part of the figure are the same (there exists a measure-preserving map between the squares), but the blue squares can also be mapped by an orthogonal map due to the Mahalanobis construction.}
\label{fig:Mahala_Wasser}
\end{figure}

\begin{figure}[htp]
\centering
\begin{tikzpicture}
\node[inner sep=0pt] (manifold) at (0,0)
    {\includegraphics[scale=0.42]{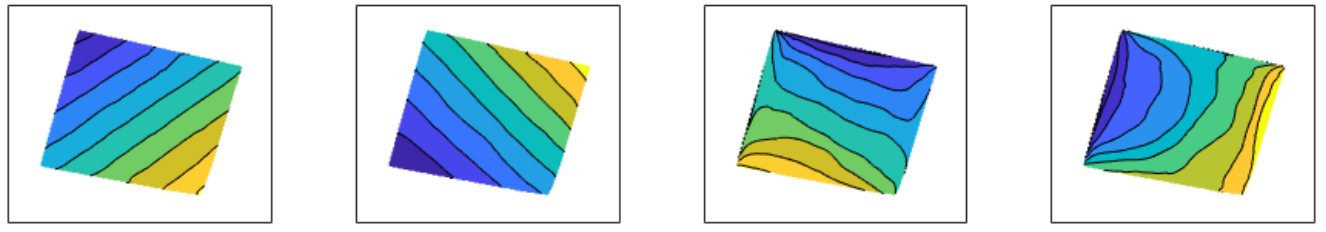}};
\node[inner sep=0pt] (phi10) at (-6,-1.5){$\phi_1$};
\node[inner sep=0pt] (phi20) at (-7.6,0){\rotatebox{90}{$\phi_2$}};
\node[inner sep=0pt] (phi12) at (-2,-1.5){$\phi_1$};
\node[inner sep=0pt] (phi23) at (-3.7,0){\rotatebox{90}{$\phi_2$}};
\node[inner sep=0pt] (phi13) at (2,-1.5){$\phi_1$};
\node[inner sep=0pt] (phi14) at (6,-1.5){$\phi_1$};
\node[inner sep=0pt] (phi20) at (0.1,0){\rotatebox{90}{$\phi_2$}};
\node[inner sep=0pt] (phi22) at (4.1,0){\rotatebox{90}{$\phi_2$}};
\node[inner sep=0pt] (cap1) at (-5.9,1.6){$(\phi_1,\phi_2)$ to $\phi_1^M$};
\node[inner sep=0pt] (cap2) at (-1.9,1.6){$(\phi_1,\phi_2)$ to $\phi_2^M$};
\node[inner sep=0pt] (cap3) at (1.9,1.6){$(\phi_1,\phi_2)$ to $\psi_1$};
\node[inner sep=0pt] (cap4) at (5.9,1.6){$(\phi_1,\phi_2)$ to $\psi_2$};
\end{tikzpicture}
\caption{Notation as in \Cref{fig:Mahala_Wasser}. We compare the functions that relate the diffusion maps embeddings. These are shown as contour plots of the respective coordinates.
\emph{First two plots:} $(\phi_1,\phi_2)$ is mapped to $(\phi^M_1,\phi^M_2)$. This function corresponds to mapping the left blue square to the right blue square in the second row of \Cref{fig:Mahala_Wasser}. It is an orthogonal map due to the isometry induced by the Mahalanobis-kernel.
\emph{Second two plots:} $(\phi_1,\phi_2)$ is mapped to $(\psi_1,\psi_2)$. This function corresponds to mapping the left blue square to the left red square in the second row of \Cref{fig:Mahala_Wasser}. This function is not an orthogonal transformation, but it is measure-preserving.}
\label{fig:func_DMAP_MAH}
\end{figure}

Note that on $\RR$, a measure-preserving map is automatically metric-preserving as well. In higher dimensions, however, there is a difference between these concepts, which is also apparent by comparing the Wasserstein and Mahalanobis frameworks. To show the difference, we consider an example from \cite{singer-2008}, which is used to explain the Mahalanobis concept. 

Consider the map 
\begin{equation}\label{eq:mushroom_map}
    S(x_1,x_2) = \left( x_1 + x_2^3, x_2 - x_1^3\right),
\end{equation}
which maps the unit square $[0,1]^2 \subset \RR^2$ to a mushroom-like subset of $\RR^2$ (therefore, as in~\cite{singer-2008}, we call $S$ the \emph{mushroom-map}), see the upper part of \Cref{fig:Mahala_Wasser}. In this figure, the blue points $x_i$ arise from a uniform sampling of the unit square, and the points $y_i$ are their images under the mushroom-map $S$ (generating ``observations'' of the points $x_i$).
If we apply DMAP with the Mahalanobis-kernel \eqref{eq:Mahala_kernel} (induced by the mushroom-map $S$, i.e.\ with Jacobians of $S^{-1}$) to the points $y_i$, we obtain an embedding denoted by $\phi^M(y_i)$: We are able to recover a (rotated) square (\Cref{fig:Mahala_Wasser} lower right). A similar square is obtained by applying standard DMAPS (i.e.\ with Euclidean distances) to the sampling points $x_i$, denoted by $\phi(x_i)$, see \Cref{fig:Mahala_Wasser} lower left.

We now map the distribution of points $y_i$ back to the uniform distribution on the square $[0,1]^2$ with Wasserstein optimal transport, denoted by $T$. Applying the mapping $T^{-1}$ pointwise to the $y_i$, we obtain the points $\tilde{x}_i$ (\Cref{fig:Mahala_Wasser} upper left, red) on the unit square $[0,1]^2$.
DMAP with Euclidean distance applied to $\tilde{x}_i$ results in the the red square described by the points $\psi(\tilde{x}_i)$ (\Cref{fig:Mahala_Wasser} lower left, red).

All three squares (blue and red in lower left and blue in lower right in \Cref{fig:Mahala_Wasser}) in DMAP space have the same distribution of points, i.e.\ can be transformed into each other with a measure preserving map (\Cref{fig:dmap_measure}). The two blue squares are even related point-wise by an orthogonal map due to the isometry induced by the Mahalanobis-kernel (\Cref{fig:dmap_orthogonal}).

We show the underlying mappings, rather than just the distributions of points, in \Cref{fig:func_DMAP_MAH}. The first row of \Cref{fig:func_DMAP_MAH} is the orthogonal mapping from $(\phi_1,\phi_2)$ to $(\phi_1^M,\phi_2^M)$ presented in the form of contour plots. The second row shows the map from $(\phi_1,\phi_2)$ to $(\psi_1,\psi_2)$ in the same manner. This map is measure-preserving, but not orthogonal.

These plots show that the Wasserstein and Mahalanobis concepts produce the same result with respect to measures, but not with respect to point-wise mappings. Wasserstein creates a point mapping from distributions; Mahalanobis creates a point mapping (and thus a way to transport distributions) via local covariances.

\section{Conclusion}
In this paper, we studied densities that arose from observations of an unknown  manifold; we focused on the case where the histograms of the observed quantities suggest the existence of singularities in the densities. Attempting to transport these histogram observations with Wasserstein optimal transport may not recover the intrinsic structure of the manifold. By assuming access to {\em additional information from an observation process}, we can employ embedding theorems to construct meaningful realizations of the underlying manifold. If enough additional information is available, through this construction, the density of points on the recovered manifold is no longer singular. 

Even in case the direct observations do provide enough information to embed the manifold, Wasserstein optimal transport may not recover the ``correct" functional relations (a fact embodied in the polar factorization theorem). To obtain such a functional relation with less ambiguity we can employ metric-preserving maps, for example constructed through diffusion maps with the Mahalanobis-metric. This approach again relies on additional process observation data. 
The reconstruction of useful pointwise maps (as opposed to mappings of distributions) and manifolds can be especially interesting and useful in contexts such as domain adaptation.

\section*{Acknowledgements}
This work was partially supported by DARPA (Lagrange Program, Drs. F. Fahroo and C. Lewis) and by the ARO through a MURI (Drs. S. Stanton and M. Munson).

\section*{Appendix}
\appendix

\section{Construction of a $C^1$ transport without additional information}\label{sec:C1_transport}

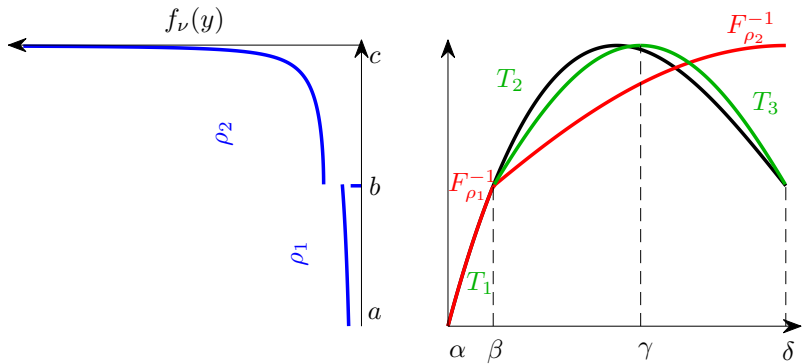
\begin{figure}
\centering
\begin{tikzpicture}  
\node[inner sep=0pt] (hook1) at (0,0)
    {\includegraphics[width=.7\textwidth]{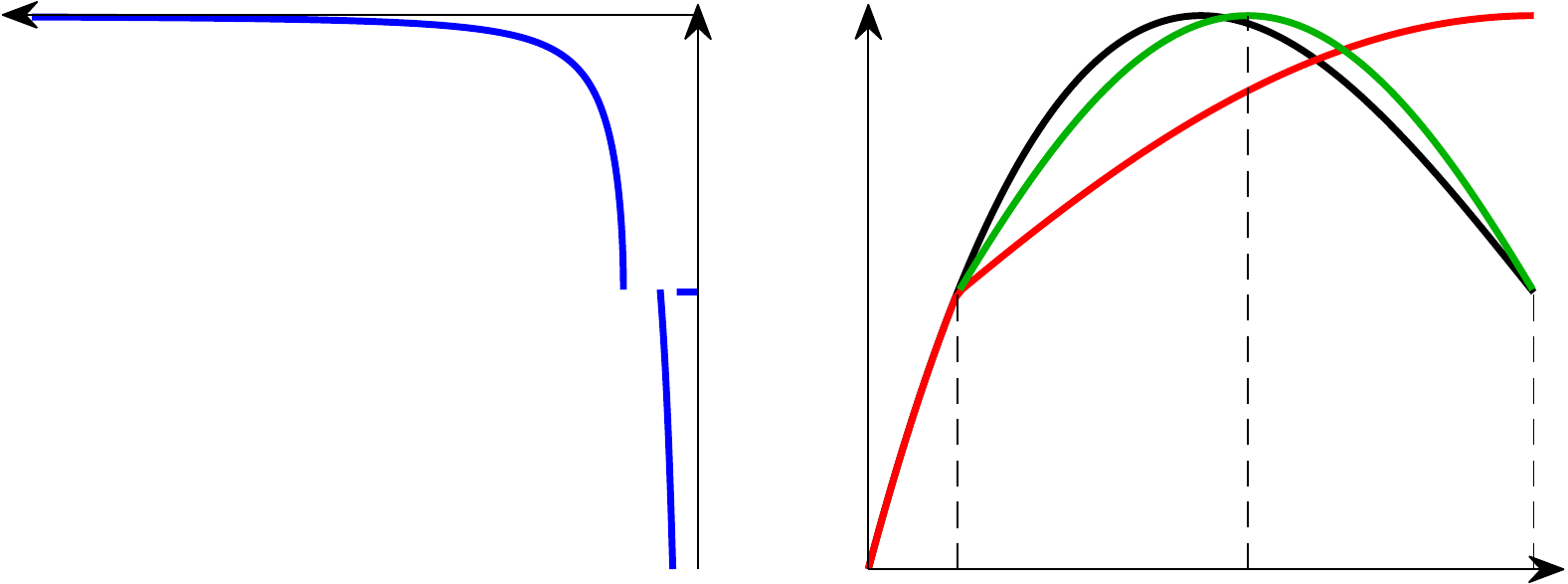}};
\node[inner sep=0pt] (alpha) at (0.7,-2.2){$\alpha$};
\node[inner sep=0pt] (beta) at (1.2,-2.2){$\beta$};
\node[inner sep=0pt] (gamma) at (3.2,-2.2){$\gamma$};
\node[inner sep=0pt] (delta) at (5.1,-2.2){$\delta$};
\node[inner sep=0pt] (F2) at (4.6,2.1){${\color{red}F_{\rho_2}^{-1}}$};
\node[inner sep=0pt] (F1) at (0.9,0){${\color{red}F_{\rho_1}^{-1}}$};
\node[inner sep=0pt] (T1) at (1,-1.3){${\color{mygreen}T_1}$};
\node[inner sep=0pt] (T2) at (1.4,1.4){${\color{mygreen}T_2}$};
\node[inner sep=0pt] (T3) at (4.8,1.1){${\color{mygreen}T_3}$};
\node[inner sep=0pt] (a) at (-0.4,-1.7){$a$};
\node[inner sep=0pt] (b) at (-0.4,0){$b$};
\node[inner sep=0pt] (c) at (-0.4,1.7){$c$};
\node[inner sep=0pt] (rho1) at (-1.4,-0.9){\rotatebox{90}{{\color{blue}$\rho_1$}}};
\node[inner sep=0pt] (rho2) at (-2.4,0.7){\rotatebox{90}{{\color{blue}$\rho_2$}}};
\node[inner sep=0pt] (rho2) at (-2.8,2.2){$f_{\nu}(y)$};
\end{tikzpicture}
\caption{Continuation of the example discussed in \Cref{fig:intro_hook}.
We construct a $C^1$ transport (shown in green) that pushes the uniform density $f_{\mu}$ on $[\alpha,\delta]$ (not shown here) to the density $f_{\nu}$ on $[a,c]$. The density $f_{\nu}$ has been obtained by pushing the uniform density by the transport $y=T(x)= -2(1-x)^3 + 1.5(1-x) + 0.5$ (black curve; we assume not to know the underlying transport).
\emph{Top left}: Discontinuous density $f_{\nu}$ consists of two continuous parts, $\rho_1$ and $\rho_2$.
\emph{Top right}: The black function is the underlying transport that pushes $f_{\mu}$ to $f_{\nu}$; we assume not to know it. The red transport is the Wasserstein optimal transport that pushes $f_{\mu}$ to $f_{\nu}$. It is computed from the cdfs $F_{\rho_1}$ and $F_{\rho_2}$, and visibly not differentiable at $\beta$. We construct a $C^1$ transport consisting of three parts $T_1,T_2$ and $T_3$ (with choice of $\gamma$ as in \eqref{eq:correct_gamma}), as indicated by the green function. The explicit construction is derived in \Cref{sec:C1_transport}. Note that $T_1$ is the same as $F_{\rho_1}^{-1}$ in this construction.}
\label{fig:hook_C1_construction}
\end{figure}

Starting only with two distributions and no additional information (e.g.\ histories), it is difficult to meaningfully reconstruct the underlying transport map. For the example of \Cref{fig:intro_hook}, we suggest an easy construction which at least produces a $C^1$ transport (in contrast to the Wasserstein transport, which is only continuous). 

The aim is to find a $C^1$ transport that pushes the uniform density on $[\alpha,\delta]$ (x-axis) to the discontinuous density $f_{\nu}$ on $[a,c]$ (y-axis), see \Cref{fig:hook_C1_construction}.
For our construction we assume that the transport we are looking for consists of three parts $T_1,T_2,T_3$, as indicated in \Cref{fig:hook_C1_construction}. Even with this restriction there exist many solutions. 

To find $T_1,T_2,T_3$, we have to solve the following problem:
\begin{equation}\label{eq:transport_rho1}
\left|\frac{d}{dy}T_1^{-1}(y)\right|=\rho_1(y), \qquad
\left|\frac{d}{dy}T_2^{-1}(y)\right|+\left|\frac{d}{dy}T_3^{-1}(y)\right|=\rho_2(y)
\end{equation}
under the constraints
\begin{equation}\label{eq:constraints}
    T_1(\beta)=T_2(\beta), \quad
    T_2(\gamma)=T_3(\gamma), \quad
    T_2'(\gamma)=T_3'(\gamma)=0.
\end{equation}
Here we use notation as in \Cref{fig:hook_C1_construction}; note that we denote the two parts of the density $f_{\nu}$ by $\rho_1$ and $\rho_2$, respectively.
The starting point $\alpha$ of our transport can be chosen freely, note however that $\delta = \alpha +1$, if we assume that the density $f_{\nu}$ integrates to $1$. 

The first problem in \eqref{eq:transport_rho1} has a unique monotonically increasing solution (given by the inverse of the cdf $F_{\rho_1}$ of $\rho_1$, compare \eqref{eq:transport_1D}).
This determines $T_1$, and also $\beta = \alpha + F_{\rho_1}(b)$.

The following choices of $T_2$ and $T_3$ (depending on the choice of the location $\gamma$ of the maximum) solve the second problem in \eqref{eq:transport_rho1} under the constraints \eqref{eq:constraints}:
\begin{equation*}
    T_2(x) = F_{\rho_2}^{-1}\left(\frac{F_{\rho_2}(c)(x-\beta)}{\gamma - \beta}\right), \quad
    T_3(x) = F_{\rho_2}^{-1}\left(\frac{F_{\rho_2}(c)(x-\beta-F_{\rho_2}(c))}{\gamma - \beta-F_{\rho_2}(c)}\right).
\end{equation*}
In this construction $T_2$ is monotonically increasing, and $T_3$ is decreasing.

The transport consisting of the three parts $T_1,T_2,T_3$ is continuous, but only the choice 
\begin{equation}\label{eq:correct_gamma}
    \gamma = F_{\rho_2}(c)\frac{\rho_1(b)}{\rho_2(b)}+\beta
\end{equation}
gives rise to a $C^1$ transport, i.e.\ satisfies $T_1'(\beta) = T_2'(\beta)$.
Note that if the original density does not have a jump discontinuity at $b$, i.e.\ $\rho_1(b)=\rho_2(b)$, then $\gamma = \delta$, i.e.\ the maximum is at the end of the interval, and our construction gives rise to the Wasserstein optimal transport (the inverse cdf as in \eqref{eq:transport_1D}). Note that only $T_1$ and $T_2$ are needed then---$T_3$ is not well-defined.

In general, our construction does not reconstruct the original transport (the green function versus the black function in \Cref{fig:hook_C1_construction}). Nevertheless it gives rise to a $C^1$ transport that pushes $f_{\mu}$ to $f_{\nu}$, in contrast to the Wasserstein transport, which is only continuous in this example.

Note that we fit the first derivatives at $\beta$ and $\gamma$ to obtain a $C^1$ transport; in general the constructed transport is not smoother than that. The second derivatives of $T_2$ and $T_3$, for example, are given by
\begin{equation*}
    T_2''(x) = - \left(\frac{F_{\rho_2}(c)}{\gamma-\beta}\right)^2 \frac{\rho_2'(T_2(x))}{\rho_2(T_2(x))^3}, \quad
    T_3''(x) = - \left(\frac{F_{\rho_2}(c)}{\gamma-\beta-F_{\rho_2}(c)}\right)^2 \frac{\rho_2'(T_2(x))}{\rho_2(T_2(x))^3},
\end{equation*}
which, in general give different values at $\gamma$. The second derivatives might be unbounded at $\gamma$, in case $\rho_2'$ grows faster than $\rho_2^3$.

In our specific example of $f_{\nu}$ (pushed from the uniform density on $[0,1]$ by the cubic polynomial $y=T(x)= -2(1-x)^3 + 1.5(1-x) + 0.5$ defined in \Cref{fig:intro_hook}), we have $c = 1$ and
\begin{equation*}
    \lim_{y \to 1^-}\frac{\rho_2'(y)}{\rho(y)^3} = \frac{3}{2}.
\end{equation*}
This implies that the transport constructed with $T_1,T_2,T_3$ has bounded second derivatives at $\gamma$, but is not $C^2$, as $T_2''(\gamma)\neq T_3''(\gamma)$.

\section{Diffusion maps (DMAP)}\label{sec:appendix dmaps}
We briefly outline the Diffusion Maps algorithm, and refer to the special case when using the Mahalanobis-metric below.
Given $N$ (possibly noisy) data points $\mathcal{D}=\{y_1,\ldots, y_N\}$ in ambient 
Euclidean space $\mathbb{E} = \mathbb{R}^{m}$  close to a smooth, compact manifold $\M$, the DMAP algorithm constructs a parametrization of $\M$ based on the convergence of the
normalized graph Laplacian on the data to the Laplace--Beltrami operator on $\M$.
First, we construct a graph between the points, where connectivity is based on a similarity measure given by a Gaussian kernel employing the Euclidean distance in the ambient space $\mathbb{E}$:
For a given scale parameter $\epsilon>0$, the similarity between two distinct points $y_i$ and $y_j$ in $\mathbb{E}$
is defined through 
$K_{ij}=k(y_i,y_j) = \rm{exp}\left(-r^2/\epsilon \right)$, where $r:=d(y_i,y_j)$. Appropriate choices of the parameter $\epsilon$ depend on the data~\cite{coifman-2006,berry-2016}.
Second, if the data points $\mathcal{D}$ are not sampled uniformly in $\M$, the matrix $K$ has to be normalized by an estimation of the density on the diagonal of a matrix $P \in \mathbb{R}^{N \times N}$, $P_{ii}=\sum_{j=1}^{N} K_{ij}$, $\widetilde{K}=P^{-\alpha}KP^{-\alpha}$ where $\alpha=0$ (no normalization, \cite{belkin-2003}) can be used in the case of uniform sampling, and $\alpha=1$ otherwise \cite{coifman-2006}.
Third, the kernel matrix $\widetilde{K}$ is normalized by the diagonal matrix $D \in \mathbb{R}^{N \times N}$, where $D_{ii}=\Sigma_{j=1}^{N} \widetilde{K}_{ij}$ for $i=1,\dots,N$.
The non-linear parametrization (embedding) of the manifold is then given by a certain number $\ell$ of eigenvectors of $A=D^{-1}\widetilde{K} \in \mathbb{R}^{N \times N}$, scaled by their respective eigenvalue (and removing redundant eigenvectors that are functions of eigenvectors associated to larger eigenvalues~\cite{dsilva-2016}).
The new embedding dimension $\ell$ may be much smaller than the ambient space dimension $m$, in which case DMAP achieves dimensionality reduction.

The similarity between points defined through the kernel can also include information about non-linear maps $S:\M\to\mathbb{R}^m$ (which have to be invertible on their image), employing the so-called ``Mahalanobis-metric'' in the kernel, first introduced in ~\cite{singer-2008}:
\begin{equation}\label{eq:Mahala_kernel}
    k(y_i,y_j)
    = \exp\left(-\frac{(y_j-y_i)^T\left(J^T(y_i) J(y_i)+J^T(y_j) J(y_j)\right) (y_j-y_i)}{2\varepsilon}\right),
\end{equation}
where $J(y)$ is the Jacobian matrix of the inverse transformation $S^{-1}$ at the point $y$.
The product $J^T J$ can be approximated through a covariance, for example, obtained by short bursts of a stochastic dynamical system with subsequent mapping by $S$~\cite{singer-2009}.

\section{Optimal transport on Riemannian manifolds}\label{sec:optimal transport on riemannian manifolds}
The results presented in \Cref{sec:notation} can be extended to Riemannian manifolds, which we summarize here, following \cite{mccann-2001}.

Let $(\M,g)$ be a smooth (for the purpose of \cite{mccann-2001}, at least $C^3$) Riemannian manifold, where $g$ denotes the metric. The volume form $\dvol_g$, given in coordinates by $\dvol_g=\sqrt{|\det(g)|}d^kx$, induces a measure on $\M$ via
$
    \vol_g(A) = \int_{\M} \ind_A \dvol_g = \int_A \dvol_g,
$    
where $A \subseteq \M$ and $\ind_A$ is the indicator function. If a measure $\mu$ on $\M$ is absolutely continuous with respect to $\vol_g$, again written as $\mu \ll \vol_g$, then there exists a \emph{density} $f_{\mu}$, such that
$
    \mu(A) = \int_A f_{\mu}(x) \dvol_g(x),
$    
with $A\subseteq \M$. This is also written as $\mu = f_{\mu} \vol_g$.

Given two measures $\mu,\nu \ll \vol_g$ on $\M$, the optimal transport problem seeks to find a smooth map $T: \M \to \M$ such that $\nu = T_{\sharp}\mu$ (where, as in \Cref{sec:transport}, $T_{\sharp}\mu(A)=\mu(T^{-1}(A))$) and such that the cost
\begin{equation}\label{eq:monge_manifold}
    \frac{1}{2}\int_{\M} d(x,T(x))^2f_{\mu}(x)\dvol_g(x),
\end{equation}
is minimized. Here $d$ denotes the Riemannian distance function on $\M$ induced by $g$. The push-forward condition $\nu = T_{\sharp}\mu$ can be replaced by \eqref{eq:bijective_transport} for bijective $T$. 
The existence of an optimal transport map is proved in \cite{mccann-2001}:
\begin{thm}\cite[Results 9--11]{mccann-2001}\label{thm:exist_OT_RM}
Let $(\M,g)$ be a connected, compact Riemannian manifold, $C^3$-smooth and without boundary. Then we have
\begin{enumerate}
 \item If $\mu\ll \vol_g$ and $\nu$ arbitrary, then there exists a smooth map $T$ satisfying $T_{\sharp}\mu = \nu$, and minimzing \eqref{eq:monge_manifold}. Only one $T$ can arise in this way (up to sets of $\mu$-measure zero). $T$ is \emph{the optimal transport pushing $\mu$ to $\nu$}.
 \item If also $\nu\ll\vol_g$, then there exists an optimal transport $T^{\ast}$ pushing $\nu$ to $\mu$. $T$ and $T^{\ast}$ are inverses of each other ($\mu$ resp.\ $\nu$ almost everywhere). 
 \item If $S: \M \to \M$ is a Borel map, $\mu$ a Radon measure, and $\nu:=S_{\sharp}\mu\ll \vol_g$. Then $S=T\circ U$, where $T$ is the optimal transport pushing $\mu$ to $\nu$ and $U$ satisfies $U_{\sharp}\mu=\mu$.
\end{enumerate}
\end{thm}


\section{Takens theorems, embedology and Whitney theorems}\label{sec:takens}

Let $k\geq \Mdim\in\mathbb{N}$, and $\mathcal{M}\subset\mathbb{R}^k$ be a $\Mdim$-dimensional, compact, smooth, connected, oriented manifold with Riemannian metric $g$ induced by its embedding in $k$-dimensional Euclidean space.
This setting is sufficient to understand the main concepts in the paper, but is more restrictive than needed for the theorems.

Together with the results from Packard et al.~\cite{packard-1980} and Aeyels~\cite{aeyels-1981}, the definitions and theorems of Takens~\cite{takens-1981} describe embedding constructions of state spaces of nonlinear dynamical systems from observations. A dynamical system is defined through its state space (here, the manifold $\mathcal{M}$) and a diffeomorphism $\phi:\mathcal{M}\to\mathcal{M}$. Here, the map $\phi$ is a discrete time dynamical system, or represents the time-$\tau$ map of a continuous-time system.
\begin{thm}[Generic delay embeddings]\label{thm:takens_1}
For pairs $(\phi,y)$, $\phi:\mathcal{M}\to\mathcal{M}$ a smooth diffeomorphism and $y:\mathcal{M}\to\mathbb{R}$ a smooth function, it is a generic property that the map $\Phi_{(\phi,y)}:\mathcal{M}\to\mathbb{R}^{2\Mdim+1}$, defined by
\begin{equation}
\Phi_{(\phi,y)}(x)=\left(y(x),y(\phi(x)),\dots,y(\underbrace{\phi\circ\dots\circ\phi}_{2\Mdim~\text{times}}(x))\right)
\end{equation}
is an embedding of $\mathcal{M}$; here, ``smooth'' means at least $C^2$.
\end{thm}
Genericity as defined by Takens~\cite{takens-1981} refers to ``an open and dense set of pairs $(\phi,y)$'' in the $C^2$ function space. In general, open and dense sets can have measure zero, so Sauer et al.~\cite{sauer-1991} later refined Takens' results significantly by introducing the concept of prevalence (a ``probability one'' analog in infinite dimensional spaces, see \cref{def:prevalence}).

Stark et al.~have extended the Takens' theorems to deterministically forced, input-output, irregularly sampled, and stochastic systems~\cite{stark-1997,stark-1999,stark-2003}. We do not discuss these results here.

The results of Sauer et al.~\cite{sauer-1991} are presented in relation to Whitney's theorems \cite{whitney-1936}.
\begin{Def}\label{def:prevalence}
A Borel subset $S$ of a normed linear space $V$ is \emph{prevalent} if there is a finite-dimensional subspace $E$ of $V$ such that for each $v\in V$, $v + e$ belongs to $S$ for (Lebesgue-) almost every $e$ in $E$.
\end{Def}

\begin{thm}[Whitney, prevalence form, embedology]\label{thm:whitney prevalence}
The set $S\subset C^1$ of smooth maps $F:\mathbb{R}^k\to\mathbb{R}^{2\Mdim+1}$ that are embeddings of $\mathcal{M}$ is prevalent.
\end{thm}
Given any smooth map $F$, there are maps arbitrarily near F that are embeddings, which is the notion of genericity from Takens. The notion of prevalence and \cref{thm:whitney prevalence} assert that ``almost all'' (in the sense of prevalence) of the maps near $F$ are embeddings.
In these theorems, the space $E$ used in the definition of prevalence is the $k(2\Mdim + 1)$-dimensional space of linear maps from $\mathbb{R}^k$ to $\mathbb{R}^{2\Mdim+ 1}$.

\section{Embedding theorems and optimal transport}\label{sec:appendix proofs}
We now prove some results concerning optimal transport in the framework of time-delay embeddings. In particular, we show that the transport maps constructed in \Cref{sec:one_dim_transport,sec:densities in r2} exist and are invertible. We argue that in our embedding constructions, optimal transport theory is applicable.
\begin{lem}\label{lem:absol_cont_measure_metric}
Let $\M$ be a smooth, orientable manifold and let $g,g'$ be two Riemannian metrics on $\M$. Then $\vol_g=f\vol_{g'}$, for smooth $f: \M \to \RR, f>0$ or $f<0$. In particular $\vol_g\ll \vol_{g'}$.
\end{lem}
\begin{proof}
On an orientable manifold, there are two possible orientations, which differ by sign.
Therefore there exists a smooth function $f: \M \to \RR$ such that $\dvol_g = f\dvol_{g'}$ and $f>0$ or $f<0$, see \cite[Chapter 15]{lee-2012}. 
\end{proof}
\begin{lem}\label{lem:diffeomorph_absol_cont}
Let $(\M,g)$ and $(\N,g')$ be two smooth, orientable Riemannian manifolds, both of dimension $\Mdim$ and let $\Phi: \M \to \N$ be a diffeomorphism. If $\mu$ is a measure on $\M$ and $\mu \ll \vol_g$, then $\Phi_{\sharp}\mu \ll \vol_{g'}$.
\end{lem}
\begin{proof}
We can push $g$ forward via $\Phi$, to obtain a metric $\Phi_{\ast}g$ on $\N$. Then \cref{lem:absol_cont_measure_metric} implies $\vol_{\Phi_{\ast}g} \ll \vol_{g'}$.
Thus
$
    \Phi_{\sharp}\mu \ll \Phi_{\sharp}\vol_g = \vol_{\Phi_{\ast}g} \ll \vol_{g'}.
$
\end{proof}
\begin{cor}
Let $(\M,g)$ be a smooth, orientable Riemannian manifold of dimension $n$. Let $F: \M \to \M $ be a smooth diffeomorphism, and let $y \in C^2(\M,\RR)$ a generic observable. Let $\Phi_{(F,y)} : \M \to \RR^{2\Mdim+1}$ given by
\begin{equation}
\Phi_{(F,y)}(p)=\left(y(p),y(F(p)),\dots,y(\underbrace{F\circ\dots\circ F}_{2\Mdim~\text{times}}(p))\right), \quad p \in \M,
\end{equation}
be the embedding from \cref{thm:takens_1}.
Let $g'$ be the Riemannian metric on $\Phi_{(F,y)}(\M)$ induced by the Euclidean metric of $\RR^{2\Mdim+1}$.
If $\mu \ll \vol_g$, then ${\Phi_{(F,y)}}_{\sharp}\mu \ll \vol_{g'}$.
\end{cor}
\begin{lem}
Let $(\M,g)$ and $(\N,g')$ be smooth (at least $C^3$), compact, connected, orientable Riemannian manifolds, without boundary and of the same dimension $n$. Let $\mu$ be a measure on $\M$ such that $\mu \ll \vol_g$. Let $\Phi_x,\Phi_y: \M \to \N$ be two diffeomorphisms. Consider the measures $\nu_x:=
{\Phi_x}_{\sharp}\mu$ and $\nu_y:={\Phi_y}_{\sharp}\mu$ on $\N$. Then $\nu_x,\nu_y\ll\vol_{g'}$. Also the optimal transport map $T: \N \to \N$ satisfying $T_{\sharp}\nu_x=\nu_y$ exists and is invertible ($\nu_x$- resp.\ $\nu_y$-a.e.).
\end{lem}
\begin{proof}
This follows from \Cref{lem:diffeomorph_absol_cont} and \cref{thm:exist_OT_RM}.
\end{proof}
\begin{cor}
Let $(\M,g)$ a be smooth (at least $C^3$), compact, connected, orientable Riemannian manifolds, without boundary, of dimension $n$.
Let $F: \M \to \M $ be a smooth ($C^2$) diffeomorphism, and let $x,y \in C^{2}(\M,\RR)$ be generic observables. Let $\Phi_{(F,x)}, \Phi_{(F,y)} : \M \to \RR^{2n+1}$
be the embedding defined in \cref{thm:takens_1}. We further assume that $\N :=\Phi_{(F,x)}(M)=\Phi_{(F,y)}(M)$. Let $\mu$ be a measure on $\M$ such that $\mu \ll \vol_g$. The optimal transport $T:\N \to\N$ that pushes ${\Phi_{(F,x)}}_{\sharp}\mu$ to ${\Phi_{(F,y)}}_{\sharp}\mu$, exists and is invertible (a.e.\ with respect to these measures).
\end{cor}
In general it is difficult to derive results on the regularity of the optimal transport map, see \cite[Chapter 12]{villani-2009}. However, for our purpose, the existence and invertability of the optimal transport map is sufficient.

\bibliographystyle{abbrvnat}


\end{document}